\documentclass[default]{sn-jnl}

\usepackage{graphicx}%
\usepackage{multirow}%
\usepackage{amsmath,amssymb,amsfonts}%
\usepackage{amsthm}%
\usepackage{mathrsfs}%
\usepackage[title]{appendix}%
\usepackage{xcolor}%
\usepackage{textcomp}%
\usepackage{manyfoot}%
\usepackage{booktabs}%
\usepackage{algorithm}%
\usepackage{algorithmicx}%
\usepackage{algpseudocode}%
\usepackage{listings}%
\usepackage{mathtools}
\usepackage{ulem}
\usepackage{bm}
\usepackage{makecell}

\theoremstyle{thmstyleone}%
\newtheorem{lemma}{Lemma}
\newtheorem{theorem}{Theorem}
\theoremstyle{thmstyletwo}%
\newtheorem{example}{Example}%
\newtheorem{remark}{Remark}%

\theoremstyle{thmstylethree}%

\newcommand{\Fp}{\mathbb{F}_p}
\newcommand{\Fq}{\mathbb{F}_q}

\DeclareMathOperator*{\Wt}{wt}
\DeclareMathOperator{\Tr}{tr}

\DeclarePairedDelimiter\nearest{\lfloor}{\rceil}
\newcommand{\Review}{}

\begin{document}

\title[Optimal $(2,\delta)$ Locally Repairable Codes via Punctured Simplex Codes]{Optimal $(2,\delta)$ Locally Repairable Codes via Punctured Simplex Codes\footnote{This paper was presented in part at the 2023 IEEE International Symposium on Information Theory (ISIT), \cite{Wang2023optimal}.}\footnote{Corresponding author: Weijun Fang (fwj@sdu.edu.cn)}}


\author[1,2]{\fnm{Yuan} \sur{Gao}}\email{gaoyuan862023@163.com}

\author[1,2,3]{\fnm{Weijun} \sur{Fang}}\email{fwj@sdu.edu.cn}
\author[4]{\fnm{Jingke} 
\sur{Xu}}\email{xujingke@sdau.edu.cn}

\author[1,2]{\fnm{Dong} \sur{Wang}}\email{d\_wang@mail.sdu.edu.cn}

\author[1,2,3]{\fnm{Sihuang} \sur{Hu}}\email{husihuang@sdu.edu.cn}

\affil[1]{\orgdiv{Key Laboratory of Cryptologic Technology and Information Security}, \orgname{Ministry of Education, Shandong University}, \orgaddress{\city{Qingdao}, \postcode{266237}, \country{China}}}


\affil[2]{\orgdiv{School of Cyber Science and Technology}, \orgname{Shandong University}, \orgaddress{\city{Qingdao}, \postcode{266237}, \country{China}}}

\affil[3]{\orgname{Quancheng Laboratory}, \orgaddress{\city{Jinan}, \postcode{250103}, \country{China}}}

\affil[4]{\orgname{School of Information Science and Engineering, Shandong Agricultural University}, \orgaddress{\city{Tai'an}, \postcode{271018}, \country{China}}}


\abstract{Locally repairable codes (LRCs) have attracted a lot of attention due to their applications
	in distributed storage systems. In this paper, we provide new constructions of optimal $(2, \delta)$-LRCs over $\Fq$ with flexible parameters. 
	Firstly, employing techniques from finite geometry, we introduce a simple yet useful condition to ensure that a punctured simplex code becomes a $(2, \delta)$-LRC. It is worth noting that this condition only imposes a requirement on the size of the puncturing set.
	Secondly, utilizing character sums over finite fields and Krawtchouk polynomials,  we determine the parameters of more punctured simplex codes with puncturing sets of new structures. Several infinite families of LRCs with new parameters are derived.
	All of our new LRCs are optimal with respect to the generalized Cadambe-Mazumdar bound and some of them are also Griesmer codes or distance-optimal codes.}



\maketitle

\section{Introduction}
\label{sec:intro}

In order to ensure the reliability of nodes in large-scale distributed
storage systems, the concept of locally repairable codes was first proposed in \cite{rLRC}.
Let $[n]=\{1,2,...,n\}$, for a linear code $\mathcal{C}$ of length $n$ over the finite field $\Fq$, a code symbol $c_i$ of $\mathcal{C}$ has locality $r$
if there exists a subset $R_i\subseteq [n]$ such that $i\in R_i,|R_i|\leq r+1$
and $c_i$ is a linear combination of $\{c_j\}_{j\in R_i\backslash\{i\}}$ over $\Fq$.
If each symbol of a codeword in $C$ has locality $r$, then $C$ is called
a locally repairable code with locality $r$ or an $r$-LRC. However, when multiple node failures happen in a distributed storage
system, the $r$-LRCs can not recover failed nodes efficiently. To address this problem, Prakash \textit{et al.} \cite{rdeltaLRC} extended the concept of $r$-LRCs to
$(r,\delta)$-LRCs which can tolerate any $\delta-1$ erasures. A code symbol $c_i$ of $C$ has locality
$(r,\delta)$ if there exists a subset $R_i\subseteq [n]$ such that $i\in R_i,|R_i|\leq r+\delta-1$
and $d(\mathcal{C}|_{R_i})\geq\delta$, where $\mathcal{C}|_{R_i}$ is the punctured code on the set $[n]\backslash R_i$.
The code $\mathcal{C}$ is called an $(r,\delta)$-LRC if all code symbols have locality $(r,\delta)$. Obviously when $\delta=2,(r,\delta)$-LRCs reduce to $r$-LRCs.

\subsection{Known Results about $(r,\delta)$-LRCs}
\label{sec:intro_known}
In \cite{rdeltaLRC}, analogous to the classical Singleton bound for general codes, the following Singleton-type bound
for an $(r,\delta)$-LRC with parameters $[n,k,d]$ is given as
\begin{equation}
	d\leq n-k+1-(\left\lceil k/r\right\rceil -1)(\delta -1).\label{singletonBound}
\end{equation}
If an $(r,\delta)$-LRC achieves the Singleton-type bound \eqref{singletonBound} with equality, then the code
is called a Singleton-optimal $(r,\delta)$-LRC. Due to its interesting
algebraic structures and 
practical applications in distributed storage systems, several constructions of Singleton-optimal $(r,\delta)$-LRCs
have been proposed in \cite{CaiHan2020bound,LJin,Guruswami2019how,TamoRSLRC,LRCMatrixProduct,BChenLRC_d5_6,WFangLRCcyclic,WFangLRCsmallLocality}. 
However, as shown in \cite[Theorem 10]{Guruswami2019how} and \cite[Theorem 2]{CaiHan2020bound}, the code length of a $q$-ary Singleton-optimal $(r,\delta)$-LRC with minimum distance $d\textgreater 2\delta$
is always upper bounded by $q^{C_{r,\delta}O(\frac{1}{d})}$, where $C_{d,r}$ is a  constant related to $r$ and $\delta$. In other words, the Singleton-type bound is no longer a suitable choice to evaluate an $(r,\delta)$-LRC with large code length over a small alphabet. The following upper bound that takes the field size into account provides a good alternative in this case.

In \cite{kupperBoundrLRC}, Cadambe and Mazumdar derived the first field-dependent bound for $q$-ary $r$-LRCs with parameters $[n,k,d]$,
\begin{equation}\label{eq:CMbound}
    k\leq\min_{\tau \in\mathbb{Z}_+}\{\tau r+k_{opt}^{(q)}(n-\tau (r+1),d)\},
\end{equation}
where \Review{$\mathbb{Z}_+$ denotes the set of all positive integers including $0$} and $k_{opt}^{(q)}(n,d)$ is the maximum dimension of a $q$-ary linear code of length $n$ and minimum distance $d$. \Review{An $r$-LRC meeting the Cadambe-Mazumdar bound \eqref{eq:CMbound} is called a $k$-optimal $r$-LRC.} In \cite{kupperBoundrLRC}, the authors proved that the binary simplex code is a $k$-optimal $2$-LRC. 
As for the $(r,\delta)$-LRCs, the generalized Cadambe-Mazumdar bound was considered in \cite{kupperBoundrDeltaLRC}. It stated \Review{that a} $q$-ary $(r,\delta)$-LRC with parameters $[n,k,d]$ has to satisfy that
\begin{equation}
	k\leq\min_{\tau\in\mathbb{Z}_+}\{\tau r+k_{opt}^{(q)}(n-\tau (r+\delta-1),d)\}.\label{CMbound_rdeltaLRC}
\end{equation}
We refer to a code that attains the generalized C-M bound \eqref{CMbound_rdeltaLRC} with equality as a $k$-optimal $(r,\delta)$-LRC.
 By deleting some columns from
the generator matrix of the simplex code, several new families of $k$-optimal LRCs with localities 2 or 3 were proposed in \cite{Anticode} and
\cite{AnticodeBin}. In \cite{LuoConstrkoptimalBin}, Luo and Cao presented several binary $k$-optimal 2-LRCs by deleting or adding some columns from a binary simplex code and used character sums to determine their parameters.
Motivated by works of \cite{optBinCodesimplicial}, Luo and Ling \cite{Luopary} constructed a family of $p$-ary
linear codes and demonstrated that they are $k$-optimal 2-LRCs in some cases. Tan \textit{et al.} \cite{minloc}
determined the locality of some known linear codes and showed that many of these codes are $k$-optimal.

\subsection{Our Contributions and Techniques}
\label{sec:intro_contri}
In this paper, we focus on new constructions of $k$-optimal $(2, \delta)$-LRCs via punctured simplex codes. We follow the construction of linear codes presented in \cite{constrLinearCode}. This construction has been applied to secret sharing schemes or LRCs \cite{sss,LuoConstrkoptimalBin,Luopary}.
We summarize the results and techniques of this paper as follows.

\begin{itemize}
\item[1.] \Review{Determining the locality is an important issue in the construction of optimal LRCs via punctured simplex codes.
In \cite{Luopary}, Luo and Ling proved the $2$-locality of a family of punctured simplex codes by explicitly finding a repair group $R_{i}$ for each code symbol $c_i$ ($i\in [n]$). The advantage of this method is that the repair groups are explicitly given. However, this method may not work when the structure of the puncturing set changes and is difficult to find potentially stronger $(2,\delta)$-locality ($\delta\geq 2$). From the perspective of finite geometry, we transform the analysis of the $(2,\delta)$-locality of a code symbol (point) of a punctured simplex code into determining whether there exists a line passing through the point in the projective space that contains enough points out of the puncturing set (see Lemma \ref{lem4}). Combining it with the Pigeonhole Principle, we give a simple but useful sufficient condition to guarantee a punctured simplex code to be a $(2,\delta)$-LRC (see Theorem \ref{suffi_condi_lrc}). It turns out that the $(2,\delta)$-locality of a punctured simplex code can be determined by the size of the puncturing set. This sufficient condition not only allows us to reveal that many known 2-LRCs via punctured simplex codes actually have stronger $(2,\delta)$-locality,
but also allows us not to focus too much on the structure of the puncturing set to ensure the locality (instead, we consider the size of the puncturing set). Consequently, we can explore more punctured simplex codes with puncturing sets of various structures. More constructions of optimal $(2,\delta)$-LRCs via punctured simplex codes are derived as follows.}
\item[2.] We generalize some results proposed by Luo and Ling \cite{Luopary} (see Theorems \ref{thm:generaliz_luo}, \ref{thm:generaliz_luo_loc} and \ref{thm:gen_luo_kopt}), Silberstein and Zeh \cite{Anticode} (see Theorem \ref{thm:weight2}), Luo and Cao \cite{LuoConstrkoptimalBin} (see Theorem \ref{thm:wt1and2}). In particular, we extend the $p$-ary linear codes presented in \cite{Luopary} to the $q$-ary linear codes, where
$p$ is a prime and $q$ is the power of $p$, and show that they have $(2,q)$-locality or $(2, q-1)$-locality.
We utilize character sums and Krawtchouk polynomials to determine the parameters of some punctured simplex codes with certain puncturing set. Specifically speaking, if the punctured
columns from the generator matrix of a simplex code have certain weight, then determining the minimum distance of the punctured simplex code
is equivalent to determining the minimum value of Krawtchouk polynomials. 
Consequently, we construct several infinite families of $k$-optimal $(2,q)$-LRCs. 
We not only generalize the results of \cite{Luopary}, \cite{Anticode} and \cite{LuoConstrkoptimalBin}, but also provide some new classes of $k$-optimal LRCs with flexible parameters. 
\end{itemize}
All our new LRCs are $k$-optimal with respect to the generalized C-M bound. Some of these codes are also Griesmer codes or distance-optimal codes.
In Table \ref{table:parameters}, we list the parameters of these codes and their optimality, where the dimensions of the codes are all $m$.
 \begin{sidewaystable}
	\caption{The parameters of some LRCs and their optimality}
	\label{table:parameters}
	\begin{tabular}{|c|c|c|c|c|c|}
		\hline
		 Code length                         &         Minimum distance            &        Locality              &    \makecell{Griesmer\\ optimality} &        $k$-optimality           & References 
		  \\ \hline 
	     $\frac{p^m-1}{p-1}-\sum_{i=1}^{t}\frac{p^{\Review{|A_i|}}-1}{p-1}$         & $p^{m-1}-\sum_{i=1}^{t}p^{\Review{|A_i|}-1}$   &  $(2,2)$ &  \checkmark  &  \checkmark & \cite[Theorem 3.1, 3.2, 4.3]{Luopary} 
	      \\ \hline 
		 $\frac{q^m-1}{q-1}-\sum_{i=1}^{t}\frac{q^{\Review{|A_i|}}-1}{q-1}$         & $q^{m-1}-\sum_{i=1}^{t}q^{\Review{|A_i|}-1}$   &  \makecell{$(2,q)$,\\$(2,q-1)$} &  \checkmark  &  \checkmark & Theorem \ref{thm:generaliz_luo}, \ref{thm:generaliz_luo_loc}, \ref{thm:gen_luo_kopt}     
		 \\ \hline  
	     $\frac{q^m-1}{q-1}-\binom{s}{2}(q-1)$    & \makecell{$q^{m-1}-\frac{(q-1)^2}{q}\binom{s}{2}+$\\ ${K_2(\nearest{s-\frac{1}{2}+\frac{1-s}{q}};s,q)}/{q}$}  & $(2,2)$  &     & \makecell{\checkmark ($2\leq q\leq 14$)}   &\cite[Theorem 10]{Anticode}    
	     \\ \hline   
		 $\frac{q^m-1}{q-1}-\binom{s}{2}(q-1)$     & \makecell{$q^{m-1}- \frac{(q-1)^2}{q}\binom{s}{2}+$\\ ${K_2(\nearest{s-\frac{1}{2}+\frac{1-s}{q}};s,q)}/{q}$}  & $(2,q)$  &     & \checkmark  &Theorem \ref{thm:weight2}      
		 \\ \hline   
		 $\frac{q^m-1}{q-1}-\binom{s}{2}(q-1)-s$        & \makecell{$q^{m-1}- \frac{(q-1)^2}{q}\binom{s}{2}-\frac{(q-1)s}{q}$\\ $+{K_1(\nearest{s-\frac{1}{2}+\frac{2-s}{q}};s,q)}/{q}$\\$+{K_2(\nearest{s-\frac{1}{2}+\frac{2-s}{q}};s,q)}/{q}$}   & $(2,q)$ &     & \checkmark  &  Theorem \ref{thm:wt1and2}   
		   \\ \hline  
	     $\frac{q^m-1}{q-1}-\frac{q^s-1}{q-1}+s$     & $q^{m-1}-q^{s-1}+1$    & $(2,q)$  & \checkmark  & \checkmark  & Theorem \ref{thm:except1}       
	       \\ \hline  
		 $\frac{q^m-1}{q-1}-\sum_{i=1}^{t}(\frac{q^{s_i}-1}{q-1}-s_i)$                & $q^{m-1}-\sum_{i=1}^{t}q^{s_i-1}+t$   & $(2,q)$ &     & \checkmark & Theorem \ref{thm:except1multi}   
		  \\ \hline   
         $\frac{3^m-1}{2}-7$ ($q=3$)              & $3^{m-1}-6$   & $(2,3)$  &     & \checkmark  &Theorem \ref{thm:s=3except2-complex}   
         \\ \hline    
		 $\frac{q^m-1}{q-1}-\frac{q^3-1}{q-1}+3(q-1)$ ($q\geq 4$)              & $q^{m-1}-q^2+2q-2$   & $(2,q)$  & \checkmark($q=4$)    & \checkmark  &Theorem \ref{thm:s=3except2-complex}    
		  \\ \hline  
		 $\frac{q^m-1}{q-1}-\frac{q^4-1}{q-1}+2(q-1)$ ($q\geq 3$)          & $q^{m-1}-q^3+q-2$   & $(2,q)$  &   \checkmark ($q=3$)  & \checkmark  & Theorem \ref{thm:s=4except2-complex}    
		 \\ \hline  
		 $\frac{q^m-1}{q-1}-\sum_{i=1}^{t}\bigg((q-1)\binom{s_i}{2}\bigg)$              & \makecell{$q^{m-1}-\frac{(q-1)^2}{q}\sum_{i=1}^t\binom{s_i}{2}$\\$+{\sum_{i=1}^{t}K_2(\nearest{s_i-\frac{1}{2}+\frac{1-s_i}{q}};s_i,q)}/{q}$}    & $(2,q)$   &      & \checkmark & Theorem \ref{thm:wt2mult}      
		  \\ \hline 
	     $\frac{q^m-1}{q-1}-\sum_{i=1}^{t}\bigg((q-1)\binom{s_i}{2}+s_i\bigg)$   & \makecell{$q^{m-1}-{\sum_{i=1}^t((q-1)^2\binom{s_i}{2}+(q-1)s_i)}/{q}$\\$+{\sum_{i=1}^{t}K_2(\nearest{s_i-\frac{1}{2}+\frac{2-s_i}{q}};s_i,q)}/{q}$\\$+{\sum_{i=1}^{t}K_2(\nearest{s_i-\frac{1}{2}+\frac{2-s_i}{q}};s_i,q)}/{q}$}    & $(2,q)$  &     & \checkmark & 	Theorem \ref{thm:wt12mult}  
	       \\ \hline
	\end{tabular}
 \end{sidewaystable}

The rest of this paper is organized as follows. In Section 2, we recall a general construction of linear codes given by Ding and Niederreiter \cite{constrLinearCode}, and some basic notations and results on finite geometry and Krawtchouk polynomials.
In Section 3, we consider $(2,\delta)$-LRCs and present several infinite families of $k$-optimal ($2,\delta$)-LRCs, some of which are also Griesmer codes or distance-optimal codes.
Section 4 concludes the paper. We remark that the parameters of all codes with explicit generator matrix presented in this paper are verified by Magma programs.

\section{Preliminaries}
\label{sec:preliminaries}
In this section, we present some preliminaries which will
be used in the subsequent sections.
Starting from now on, we adopt the following notations:
\begin{itemize}
	\item  We use $\left\lfloor \cdot \right\rfloor$ and $\left\lceil \cdot \right\rceil$ to denote the floor function and the ceiling function, respectively.
	For $x\in \mathbb{R}$,  we define 
 \begin{align}\label{def:nearest}
		\nearest{x}:=\begin{cases}
			\left\lfloor x\right\rfloor, &\text{ if }x-\left\lfloor x \right\rfloor\leq 0.5,\\
			\left\lceil x\right\rceil, & \text{ otherwise }.
		\end{cases}
\end{align}
	\item We use $[a,b]$ to denote the closed interval $\{x\in \mathbb{R}:a\leq x\leq b\}$, where $a\leq b\in \mathbb{R}$.
	\item Let $m$ be a positive integer. We use $[m]$ to denote the set $\{1, 2, \cdots, m\}$.
 \item For some $\bm{x}=(x_1,x_2,\dots,x_m)\in \Fq^{m}$ and non-empty set $A=\{i_1,i_2,\dots,i_s\}\subseteq [m]$, where $i_1\textless i_2\textless\cdots\textless i_s$, we use $\bm{x}_{A}$ to denote the projection of $\bm{x}$ on $A$, i.e., 
 \begin{align}\label{def:projection}
		\bm{x}_A:={(x_{i_1},x_{i_2},\dots,x_{i_s})}.
	\end{align}
\end{itemize}
\subsection{A General Construction of Linear Codes}
\label{sec:pre_constr_lineaer_codes}
In this subsection, we describe a general construction of linear code which was given by Ding and Niederreiter \cite{constrLinearCode}. 
Let $m$ be a positive integer, $q$ a power of some prime $p$,
$\mathbb{F} _{q}$ the finite field containing $q$ elements and $\Fq^m$ the vector space over $\Fq$ of dimension $m$. For any vector $\bm{x}$=
$(x_1,x_2,\cdots,x_m)\in\Fq^m$, the Hamming weight of $\bm{x}$ is given as wt$(\bm{x})=|\{1\leq i\leq m:x_i\neq 0\}|$. We let tr$_{q^m/q}(\cdot)$ be the trace function from $\mathbb{F}_{q^m}$ to $\Fq$ and tr$(\cdot)$ the absolute trace function from $\Fq$ to $\Fp$.

Ding and Niederreiter \cite{constrLinearCode} established a general construction of linear codes, which says that if $D$ is a proper subset of $\mathbb{F} _{q^m}$ and denote $D^c=\mathbb{F} _{q^m}\backslash D=\{d_1,d_2,\cdots,d_n\}$,
a $q$-ary linear code of length $n$ is constructed by
\begin{equation}
	\mathcal{C}_{D^c}=\{\bm{c}_x=(\Tr_{q^m/q}(xd_1),\cdots,\Tr_{q^m/q}(xd_n)):x\in\mathbb{F} _{q^m}\}.\label{gen_constr}
\end{equation}
If $D$ is a proper subset of $\mathbb{F}^m_{q}$ and denote $D^c=\mathbb{F}^m_{q}\backslash D=\{\bm d_1, \bm d_2, \cdots, \bm d_n\}$, then the above construction (\ref{gen_constr}) can be modified to
\begin{equation}
	\mathcal{C}_{D^c}=\{\bm{c_x}=(\bm{x\cdot d}_1,...,\bm{x\cdot d}_n):\bm{x}\in\Fq^m\},\label{gen_linear_constr}
\end{equation}
where $\bm{x\cdot d}_i$ is the Euclidean inner product of $\bm{x}$ and $\bm{d}_i$. 

The following two bounds are useful in subsequent sections.
\begin{lemma}[\cite{IntroToCoding}, Griesmer Bound]
	Let $C$ be a $q$-ary $[n,k,d]$ linear code, then 
 \begin{align*}
     n\geq\sum_{i=0}^{k-1}\left\lceil \frac{d}{q^i}\right\rceil.
  \end{align*}
\end{lemma}
A $q$-ary linear code achieving the Griesmer bound with equality is called a Griesmer code. A $q$-ary [n,k,d] linear code is called distance-optimal if no  $q$-ary $[n,k,d+1]$ linear code exists.
\begin{lemma}[\cite{IntroToCoding}, Plotkin Bound]
	Let $\mathcal{C}$ be a $q$-ary code with $M$ codewords, length $n$ and minimum distance $d$. If $qd>(q-1)n$, then 
\begin{align*}
 M\leq\frac{qd}{qd-(q-1)n}.
 \end{align*}
\end{lemma}
\subsection{Finite Geometry}
\label{sec:pre_finite_geo}
The projective space
$PG(m-1, q)$ over $\mathbb{F}_q$ is the geometry whose points, lines, planes, $\cdots$ , hyperplanes are the
subspaces of $\mathbb{F}^{m}_q$ of dimension $1, 2, 3, \cdots , m-1$.  So, we also use a nonzero vector $\bm g \in \mathbb{F}^{m}_q$ to denote the point in $PG(m-1, q)$. Two nonzero vectors $\bm g_1$ and $\bm g_2$ are the same point in $PG(m-1, q)$ if and only if $\bm g_1=\lambda \bm g_2$ for some $\lambda \in \mathbb{F}^*_q$. The points $\bm{g}_1, \bm g_2, \cdots,\bm{g}_{\delta+1}$ are called collinear if they lie on a common \Review{line}, i.e., any three of them are linearly dependent.  Note that when we replace \Review{$\bm d_i$} by \Review{$\lambda \bm d_i$} for any $\lambda \in \mathbb{F}^*_q$, the parameters of the code given by Eq. \eqref{gen_linear_constr} do not change. So we rewrite the code construction given in Eq. \eqref{gen_linear_constr} via the language of projective geometry as follows. Suppose $D\subset PG(m-1,q)$ is a proper subset of $PG(m-1,q)$ and $D^c=PG(m-1,q)\backslash D=\{\bm{d}_1,\bm{d}_2,\cdots,\bm{d}_{n}\}$, then a $q$-ary linear code of length $n$ is constructed by
\begin{equation}
	\mathcal{C}_{D^c}=\{\bm{c_x}=(\bm{x\cdot d}_1,\cdots,\bm{x\cdot d}_n):\bm{x}\in\Fq^m\}.\label{eq:CodeDef}
\end{equation} 
In this paper, we will use Eq. \eqref{eq:CodeDef} to construct optimal LRCs. 
Note that when $D=\varnothing$, $\mathcal{C}_{D^c}$ is the famous simplex code. Thus in this sense, for general proper subset $D$ of $PG(m-1,q)$, the code $\mathcal{C}_{D^c}$ is the punctured code of the simplex code on set $D$. 

Using character sums over finite fields, we can compute the parameters of those constructed codes.
Assume that $\omega_p$ is
the primitive $p^{\text{th}}$ root of unity in the complex number field $\mathbb{C}$, then for $a\in\Fq$, the additive character $\chi_a$ from $\Fq$ to $\Fp$ is defined as
$\chi_a(c)=\omega_p^{\Tr(ac)}, \text{for all }c\in\Fq.$
If $a=0$, then $\sum_{c\in\Fq}\chi_a(c)=q$; otherwise $\sum_{c\in\Fq}\chi_a(c)=0$ (\cite{NiedFF}). Next, we introduce how to determine the minimum distance of $\mathcal{C}_{D^c}$, which will be applied frequently in this paper. 
	Let $\bm{x}=(x_1,\cdots,x_m)$ be any nonzero vector of $\Fq^m$, then
	\begin{align}
		\Wt(\bm{c_x})&=|D^c|-|\{\bm{d}\in D^c:\bm{x\cdot d=0}\}|\notag \\
		&=|D^c|-\sum_{\bm{d}\in D^c}\frac{1}{q}\sum_{y\in\Fq}\omega_p^{\Tr(y\bm{x\cdot d})}\notag\\
		&=\frac{q-1}{q}|D^c|-\frac{1}{q}\sum_{\bm{d}\in D^c}\sum_{y\in\Fq^*}\omega_p^{\Tr(y\bm{x\cdot d})}\notag\\
		&=\frac{q-1}{q}|D^c|-\frac{1}{q}\sum_{\bm{d}\in (\Fq^m)^*}\omega_p^{\Tr(\bm{x\cdot d})}+\frac{1}{q}\sum_{\bm{d}\in D}\sum_{y\in\Fq^*}\omega_p^{\Tr(y\bm{x\cdot d})}.
	\end{align}
	Denote $\bm d=(d_1,\cdots,d_m)$,
 then 
 \begin{align*}
		\sum_{\bm{d}\in (\Fq^m)^*}\omega_p^{\Tr(\bm{x\cdot d})}&=\sum_{d_1\in\Fq}\cdots\sum_{d_m\in\Fq}\omega_p^{\Tr(x_1d_1)}\cdots\omega_p^{\Tr(x_md_m)}-1\\
		&=\prod_{i=1}^m\left(\sum_{d_i\in\Fq}\omega_p^{\Tr(x_id_i)}\right)-1=-1.
	\end{align*}
 As such we have 
\begin{align}
    \min_{\bm{x}\in (\Fq^{m})^*} \Wt(\bm{c_x})=\frac{q-1}{q}|D^c|+\frac{1}{q}+\frac{1}{q} \min_{\bm{x}\in (\Fq^{m})^*} \left\{\sum_{\bm{d}\in D}\sum_{y\in {\Fq}^*}\omega_p^{\Tr(y\bm{x\cdot d})}\right\}. \label{eq:calcu_minimum_weight}
\end{align}
Therefore, to determine the value $\min\limits_{\bm{x}\in (\Fq^{m})^*} \Wt(\bm{c_x})$, we only need to determine the value of $\min\limits_{\bm{x}\in (\Fq^{m})^*}\{\sum\limits_{\bm{d}\in D}\sum\limits_{y\in {\Fq}^*}\omega_p^{\Tr(y\bm{x\cdot d})}\}$. In particular, if $\min\limits_{\bm{x}\in (\Fq^{m})^*}\Wt(\bm{c_x}) \textgreater 0$, then we have $\dim(\mathcal{C}_{D^c})=m$ and $d(\mathcal{C}_{D^c})=\min\limits_{\bm{x}\in (\Fq^{m})^*}\Wt(\bm{c_x})$.
 
 We let the points in $PG(m-1,q)$ be the vectors in $\Fq^m$ that the first nonzero coordinate is 1 for simplicity. If $A$ is a nonempty subset of $[m]$, we let $P_{[m]}=PG(m-1,q)$ and $P_A$ be the subset of $PG(m-1,q)$ consisting of points whose coordinates outside of $A$ are all 0. It is easy to see that 
$$|P_A|=\frac{q^{|A|}-1}{q-1},~\bigcup\limits_{\alpha\in\Fq^*}\alpha P_A=L_A^*,$$ where
$\alpha P_A=\{\alpha\bm{a}:\bm{a}\in P_A\}$ and
\begin{equation}
	L_A=\{(a_1,\cdots,a_m)\in\mathbb{F} _q^m:a_i=0\text{ if } i\notin A\}.\label{projec_space}
\end{equation}
For any two subsets $A_1,A_2$ of
$[m]$, the intersection of $P_{A_1}$ and $P_{A_2}$ is equal to $P_{A_1\cap A_2}$, where $P_\varnothing  =\varnothing.$
\subsection{Krawtchouk Polynomials}
\label{sec:pre_Kraw_poly}
In this subsection, we briefly review some basic results of Krawtchouk polynomials.

Given positive integers $n,q$, and suppose $0 \leq k \leq n$, the Krawtchouk polynomial of degree $k$ is defined as \cite{IntroToCoding} 
\begin{align*}
K_k(x;n,q):=K_k(x):=\sum_{j=0}^k(-1)^j\binom{x}{j}\binom{n-x}{k-j}(q-1)^{k-j},
\end{align*}
where $\binom{x}{j}:=\frac{x(x-1)\cdots(x-j+1)}{j!} \text{ for } x\in \mathbb{R}$.

The following lemma is a slight modification of \cite[Lemma(5.3.1)]{IntroToCoding}. It will be applied frequently to determine the minimum distance of the codes.
\begin{lemma}[\cite{IntroToCoding}, Lemma(5.3.1)]
	Let $a,s,w$ be integers satisfying $s\geq 1$, $0\leq a\leq  s$, $1\leq w\leq s$ and $\bm{x}$ be a vector of length $s$ over $\Fq$ with wt$(\bm{x})=a$. Then
	we have $$\sum_{\bm{y}\in\Fq^s,\Wt(\bm{y})=w}\omega _p^{\Tr(\bm{x}\cdot \bm{y})}=K_w(a;s,q).$$\label{kraw}
\end{lemma}
\section{$(2,\delta)$-LRCs from Punctured Simplex Codes}
\label{sec:main_results}
In this section, we will provide several constructions of LRCs via punctured simplex codes. Firstly, we give a simple lemma which will be used to determine the locality of linear codes. 
\begin{lemma}\label{lem4}
	Let $\delta\geq 2$ be an integer, $\bm{g}_1, \bm g_2, \cdots,\bm{g}_{\delta+1}$ be $\delta+1$ distinct collinear points in $PG(m-1,q)$. Let $\mathcal{C}$ be the linear code with the generator matrix $G=[\bm{g}_1\ \bm g_2\ \cdots\ \bm{g}_{\delta+1}]$, then $\mathcal{C}$ is a $q$-ary $[\delta+1, 2, \delta]$-MDS code.\label{delta1line}
\end{lemma}
\begin{proof}
	Since any two of $\bm{g}_1,\bm g_2, \cdots,\bm{g}_{\delta+1}$ are linearly independent and any three of $\bm{g}_1, \bm g_2, \cdots,\bm{g}_{\delta+1}$ are linearly dependent, we have $rank(G)=2$. Thus $\dim(\mathcal{C})=2$ and $\dim(\mathcal{C}^{\perp})=\delta-1$. On the other hand, $G$ is the parity-check matrix of $\mathcal{C}^{\perp}$, then $d(\mathcal{C}^{\perp}) \geq 3$ since any two of $\bm{g}_1, \bm g_2, \cdots,\bm{g}_{\delta+1}$ are linearly independent. By the Singleton bound, $d(C^{\perp}) \leq \delta+1-(\delta-1)+1=3$, hence $\mathcal{C}^{\perp}$ is a $[\delta+1,\delta-1,3]$-MDS code. So $\mathcal{C}$ is a $[\delta+1,2,\delta]$-MDS code.
\end{proof}

In the following, we present a sufficient condition that guarantees a punctured simplex code to be a $(2,\delta)$-LRC.
\begin{theorem}\label{suffi_condi_lrc}
	Suppose $2\leq \delta \leq q$ and  $D$ is a subset of $PG(m-1, q)$. If $|D|\leq\frac{q^{m-1}-1}{q-1}(q+1-\delta)-1$, then the code $\mathcal{C}_{D^c}$ given in Eq. \eqref{eq:CodeDef} is a $q$-ary $(2,\delta)$-LRC, where $D^c=PG(m-1,q)\setminus D $.
\end{theorem}
\begin{proof}
	For any point $\bm{g}\in D^c$, there are $\frac{q^{m-1}-1}{q-1}$ lines in $PG(m-1,q)$ containing $\bm{g}$, and each line has $q+1$ points. Since $|D|\leq\frac{q^{m-1}-1}{q-1}(q+1-\delta)-1$, by the Pigeonhole Principle, there exists at least one line $L$ containing $\bm g$, such that there are $\delta+1$ points $\bm g_1=\bm g, \bm g_2, \cdots, \bm g_{\delta+1}$ of $L$ belonging to the subset $D^c$. By Lemma \ref{lem4}, $d((\mathcal{C}_{D^c})_{|E})=\delta$, where $E=\{\bm g_1,\bm g_2, \cdots, \bm g_{\delta+1}\}$. Hence the code $\mathcal{C}_{D^c}$ has $(2, \delta)$-locality.
\end{proof}
\begin{remark}
	When $D=\varnothing$, then $\mathcal{C}_{D^c}$ is just  the $q$-ary simplex code. From Theorem \ref{suffi_condi_lrc}, we know that the $q$-ary simplex codes have locality $(2,q)$. 
	In particular, to ensure the code $\mathcal{C}_{D^c}$ to be a $2$-LRC, it only needs to satisfy that $|D| \leq q^{m-1}-2$. By using this sufficient condition, 
 the $2$-locality of the $2$-LRCs proposed in \cite{Luopary}
 can be determined simpler.
\end{remark}
Hyun \textit{et al.} \cite{optBinCodesimplicial} constructed infinite families of binary Griesmer codes punctured by unions of projective spaces. Luo and Ling \cite{Luopary} obtained similar results of linear codes over $\mathbb{F}_p$. In the following, we extend their results to general $q$-ary codes.
\begin{theorem}\label{thm:generaliz_luo}
	Let $m,t>1$ be positive integers. Assume that $A_1,...,A_t$ are nonempty subsets of $[m]$ satisfying $|A_1|\leq |A_2|\leq \dots\leq |A_t|$ and $A_i\cap A_j=\varnothing $ for any $i\neq j\in[t].$
	Let $D=\bigcup\limits_{i=1}^tP_{A_i}$ and $D^c=P_{[m]}\backslash D$, then the code $\mathcal{C}_{D^c}$ defined by Eq. \eqref{eq:CodeDef} is a $q$-ary linear code with parameters $[\frac{q^m-1}{q-1}-\frac{\sum_{i=1}^tq^{|A_i|}-t}{q-1},m,q^{m-1}-\sum_{i=1}^tq^{|A_i|-1}]$. Furthermore, assume that $|A_1|=...=|A_{i_1}|=s_1,|A_{i_1+1}|=...=|A_{i_2}|=s_2,
	...,|A_{i_{u-1}+1}|=...=|A_{i_u}|=s_u$, where $s_1\textless s_2\textless ...\textless s_u$ and $i_u=t$. If $\max\{i_1,i_2-i_1,...,i_u-i_{u-1}\}\leq q-1$, then $\mathcal{C}_{D^c}$ is a Griesmer code. 
\end{theorem}
\begin{proof}
	Note that $P_{A_i}\cap P_{A_j}=P_{A_i\cap A_j}=\varnothing$ for any $i\neq j\in[t]$, so we have $|D|=\sum_{i=1}^t|P_{A_i}|=\frac{\sum_{i=1}^tq^{|A_i|}-t}{q-1}$, thus the length of $\mathcal{C}_{D^c}$ is $n=\frac{q^m-1}{q-1}-\frac{\sum_{i=1}^tq^{|A_i|}-t}{q-1}$.
	According to Eq. \eqref{eq:calcu_minimum_weight}, to determine the value of $\min\limits_{\bm{x}\in(\Fq^{m})^*}\Wt(\bm{c_x})$, we only need to determine the value of $\min\limits_{\bm{x}\in (\Fq^{m})^*}\{\sum\limits_{\bm{d}\in D}\sum\limits_{y\in\Fq^*}\omega_p^{\Tr(y\bm{x\cdot d})}\}$. For any $\bm{x}\in ({\Fq^m})^*$, we have
\begin{align*}
		\sum_{\bm{d}\in D}\sum_{y\in\Fq^*}\omega_p^{\Tr(y\bm{x\cdot d})}=\sum_{i=1}^t\sum_{\bm{d}\in P_{A_i}}\sum_{y\in\Fq^*}\omega_p^{\Tr(\bm{x\cdot}(y\bm{d}))}=\sum_{i=1}^t\sum_{\bm{d}\in (\Fq^{|A_i|})^*}\omega_p^{\Tr(\bm{x}_{A_i}\cdot \bm{d})},
	\end{align*}
 where $\bm x_{A_i}$ is the projection vector of $\bm x$ on $A_i$ (see Eq. \eqref{def:projection}).
	Note that
 \begin{align*}
		\sum_{\bm{d}\in (\Fq^{|A_i|})^*}\omega_p^{\Tr(\bm{x}_{A_i}\cdot\bm{d})}=\begin{cases}
			q^{|A_i|}-1, &\textnormal{ if }\bm{x}_{A_i}=\bm{0},\\
			-1, &\textnormal{ if } \bm{x}_{A_i}\neq\bm{0}.
		\end{cases}
	\end{align*}
We have 
$\min\limits_{\bm{x}\in (\Fq^{m})^*}\{\sum\limits_{\bm{d}\in D}\sum\limits_{y\in\Fq^*}\omega_p^{\Tr(y\bm{x\cdot d})}\}=-t.$
 By Eq. \eqref{eq:calcu_minimum_weight}, the minimum weight of $\mathcal{C}_{D^c}$ is 
 \begin{align*}
 d:=\min\limits_{\bm{x}\in(\Fq^{m})^*}\Wt(\bm{c_x})=\frac{q-1}{q}|D^c|+\frac{1}{q}-\frac{t}{q}
	=q^{m-1}-\sum\limits_{i=1}^tq^{|A_i|-1}.
 \end{align*}
	It is easy to prove that $q^{m}-\sum_{i=1}^tq^{|A_i|}>0$ since $\sum_{i=1}^t |A_i|\leq m$. Thus wt$(\bm {c_x})=0$ if and only if $\bm x=\bm 0$, hence the dimension and minimum distance of $\mathcal{C}_{D^c}$ are $m$ and $d$, respectively.
	
	  Suppose $\sum_{i=1}^tq^{|A_i|-1}=\sum_{i=g}^hb_iq^i$, where $0\leq b_i\leq q-1 \text{ for } i=g,...,h$. Then
	\begin{align*}
	&\sum_{i=0}^{m-1}\left\lceil\frac{q^{m-1}-\sum_{j=1}^tq^{|A_j|-1}}{q^i}\right\rceil=\sum_{i=0}^{m-1}\left\lceil\frac{q^{m-1}-\sum_{j=g}^hb_jq^j}{q^i}\right\rceil\\
	&\Review{=\sum_{i=0}^{m-1}q^{m-1-i}-\sum_{j=g}^h b_j\sum_{i=0}^{j}q^{j-i}}\\
	&\Review{=\sum_{i=0}^{m-1}q^{m-1-i}-\sum_{j=g}^h b_j\frac{q^{j+1}-1}{q-1}}\\
	&=\frac{q^m-1}{q-1}-\frac{\sum_{i=1}^tq^{|A_i|}-\sum_{i=g}^hb_i}{q-1}.
\end{align*}
	As $\max\{i_1,i_2-i_1,...,i_u-i_{u-1}\}\leq q-1,\sum_{i=g}^hb_i=i_1+i_2-i_1+...+i_u-i_{u-1}=i_u=t.$ The length of $\mathcal{C}_{D^c}$ is $\frac{q^m-1}{q-1}-|D|=\frac{q^m-1}{q-1}-\frac{\sum_{i=1}^tq^{|A_i|}-\sum_{i=g}^hb_i}{q-1},$ hence the code $\mathcal{C}_{D^c}$ is a Griesmer code.
\end{proof}
We now investigate the locality of the codes given in Theorem \ref{thm:generaliz_luo}.
\begin{theorem}\label{thm:generaliz_luo_loc}
	Keep the notation as in Theorem \ref{thm:generaliz_luo}. If $t=2$ and $|A_i|\leq m-2$ for all $i\in[t]$, then the code $\mathcal{C}_{D^c}$ has locality $(2,q)$; if $t\geq 3$ and $m\geq 4$, then the code $\mathcal{C}_{D^c}$ has locality $(2,q)$; if $m>t=2,q>2$ and $|A_2|=m-1$, then the code $\mathcal{C}_{D^c}$ has locality $(2,q-1)$.
\end{theorem}
\begin{proof}
	\textbf{Case 1:} $t=2, |A_i|\leq m-2,i=1,2.$
	If $m\geq 4$, then $|D|=\frac{q^{|A_1|}+q^{|A_2|}-2}{q-1}\leq\frac{q^2+q^{m-2}-2}{q-1}\leq\frac{q^{m-1}-q}{q-1}$; if $m=3$, then
	$|D|=\frac{q^{|A_1|}+q^{|A_2|}-2}{q-1}=2\leq\frac{q^2-q}{q-1}$. By Theorem \ref{suffi_condi_lrc}, we can deduce that the code $\mathcal{C}_{D^c}$ has locality $(2,q)$.
	
	
	\textbf{Case 2:} $t\geq 3$ and $m\geq 4$. Note that
	$|D|=\frac{\sum_{i=1}^tq^{|A_i|}-t}{q-1}\leq\frac{q^{m-t+1}+(t-1)q-t}{q-1}=\frac{q^{m-t+1}+t(q-1)-q}{q-1}\leq\frac{q^{m-t+1}+m(q-1)-q}{q-1}\leq \frac{q^{m-2}+q^{m-2}(q-1)-q}{q-1}=\frac{q^{m-1}-q}{q-1}$. By Theorem \ref{suffi_condi_lrc}, we can deduce that the code $\mathcal{C}_{D^c}$ has locality $(2,q)$.
	
	
	\textbf{Case 3:} $m>t=2, |A_1|=1,|A_2|=m-1,q>2$. Note that $|D|=\frac{q^{m-1}+q-2}{q-1}\leq\frac{2q^{m-1}-2}{q-1}-1$. By Theorem \ref{suffi_condi_lrc}, we can deduce that the code $\mathcal{C}_{D^c}$ has locality $(2,q-1)$.
\end{proof}
 
\begin{theorem}\label{thm:gen_luo_kopt}
	Keep the notation  as in Theorems \ref{thm:generaliz_luo} and \ref{thm:generaliz_luo_loc}, the code $\mathcal{C}_{D^c}$ is \Review{a} $k$-optimal LRC with respect to the bound \eqref{CMbound_rdeltaLRC}. 
\end{theorem}
\begin{proof} We use $[n,k=m,d]$ to denote the parameters of the code $\mathcal{C}_{D^c}$.

	\textbf{Case 1:} $t=2, |A_i|\leq m-2,i=1,2.$ By Theorem \ref{thm:generaliz_luo_loc}, the code $\mathcal{C}_{D^c}$ is a $(2,q)$-LRC. Since it is a Griesmer code, we have $n=\sum_{i=0}^{m-1}\left\lceil \frac{d}{q^i}\right\rceil$. Note that $\left\lceil \frac{d}{q^{m-1}}\right\rceil=1$. Thus, $\sum_{i=0}^{m-2}\left\lceil \frac{d}{q^i}\right\rceil=n-1\textgreater n-q-1$. Then we have $k_{opt}^{(q)}(n-q-1,d)\leq m-2$ by the Griesmer bound. By Eq. \eqref{CMbound_rdeltaLRC}, $k\leq 2+k_{opt}^{(q)}(n-1-q,d)\leq m$ and the code $\mathcal{C}_{D^c}$ is $k$-optimal.

	\textbf{Case 2:} $t\geq 3$ and $m\geq 4$. Its proof is similar to that of \textbf{Case 1}.
	
	\textbf{Case 3:} $m>t=2, |A_1|=1,|A_2|=m-1,q>2$. By Theorem \ref{thm:generaliz_luo_loc}, the code $\mathcal{C}_{D^c}$ is a $(2,q-1)$-LRC. Since it is a Griesmer code, we have $n=\sum_{i=0}^{m-1}\left\lceil \frac{d}{q^i}\right\rceil$. Note that $\left\lceil \frac{d}{q^{m-1}}\right\rceil=1$. Thus, $\sum_{i=0}^{m-2}\left\lceil \frac{d}{q^i}\right\rceil=n-1\textgreater n-q$. Then we have $k_{opt}^{(q)}(n-q,d)\leq m-2$ by the Griesmer bound. By Eq. \eqref{CMbound_rdeltaLRC}, $k\leq 2+k_{opt}^{(q)}(n-q,d)\leq m$ and the code $\mathcal{C}_{D^c}$ is $k$-optimal.
\end{proof}
\begin{remark}
 In \cite[Theorem 3.1, 3.2, 4.3]{Luopary}, Luo and Ling proved that the codes $\mathcal{C}_{D^c}$ in Theorem \ref{thm:generaliz_luo} are both Griesmer codes and $k$-optimal $2$-LRCs while $q$ is an odd prime. In Theorem \ref{thm:generaliz_luo}, \ref{thm:generaliz_luo_loc} and \ref{thm:gen_luo_kopt}, we generalize their results from the following perspectives:
 \begin{itemize}
     \item Our constructions allow for all finite field $\mathbb{F}_q$.
     \item \cite[Theorem 4.3]{Luopary} shows that the codes in Theorem \ref{thm:generaliz_luo} (for $q$ is an odd prime) are $k$-optimal $2$-LRCs.
     However, in Theorem \ref{thm:generaliz_luo_loc} and \ref{thm:gen_luo_kopt}, we show that the codes in Theorem \ref{thm:generaliz_luo} are actually $k$-optimal $(2,q)$-LRCs or $(2,q-1)$-LRCs (except the case ``$q=2$, $t=2$ and $|A_2|=m-1$''). 
 \end{itemize}
\end{remark}
\begin{example}
	Let $q=4,m=3$ and $A_1=\{1\},A_2=\{2,3\}$, then $\mathcal{C}_{D^c}$ defined in Theorem \ref{thm:generaliz_luo} is a 4-ary $[15,3,11]$ Griesmer code with a generator matrix 
 \begin{align*}
 G=\begin{pmatrix}
		1       &1       &1       &1       &1       &1       &1       &1       &1       &1       &1       &1       &1       &1       &1       \\
		1       &\alpha  &\alpha+1&0       &1       &\alpha  &\alpha+1  &0       &1       &\alpha  &\alpha+1&0       &1       &\alpha  &\alpha+1\\
		0       &0       &0       &1       &1       &1       &1     &  \alpha  &\alpha  &\alpha  &\alpha  &\alpha+1&\alpha+1&\alpha+1&\alpha+1
	\end{pmatrix},\end{align*} where $\alpha$ is a primitive element in $\mathbb{F}_4$. $\mathcal{C}_{D^c}$ is a $(2, 3)$-LRC. For instance, one can see that the columns $(1,1,0)^T,(1,0,1)^T,(1,\alpha,\alpha+1)^T,(1,\alpha+1,\alpha)^T$
		of $G$ generate a $[4,2,3]$-code. Hence the first symbol of $\mathcal{C}_{D^c}$ has locality $(2,3)$. 
		Note that $k^{(4)}_{opt}(11,11)=1$, thus $\mathcal{C}_{D^c}$ attains the generalized C-M bound  \eqref{CMbound_rdeltaLRC}.
	\end{example}
	In the following, we consider another family of punctured simplex codes, which is motivated by \cite[Theorem 10]{Anticode}.
	\begin{theorem}\label{thm:weight2}
		Let $m\geq 3$, $A\subseteq [m]$ with $|A|=s \geq 2$, $D=\{\bm{d}\in P_A: \Wt(\bm{d})=2\}$ and $D^c=P_{[m]}\backslash D$. Then the code $\mathcal{C}_{D^c}$ defined in Eq. \eqref{eq:CodeDef} is a $q$-ary $k$-optimal $(2,q)$-LRC with parameters $$[n,k,d]=\left[\frac{q^m-1}{q-1}-(q-1)\binom{s}{2},m,\frac{(q-1)n+1}{q}+\frac{\Delta}{q}\right]$$ provided that $(q-1)\binom{s}{2}\leq\frac{q^{m-1}-q}{q-1}$ and $0\textless \frac{qd}{\Delta+q^2}\textless q^{m-1}$, where $\Delta=K_2(\nearest{s-\frac{1}{2}+\frac{1-s}{q}};s,q)$ (the function $\nearest{\cdot}$ is defined by Eq. \eqref{def:nearest}).
	\end{theorem}
	\begin{proof}
		There are $(q-1)^2\binom{s}{2}$ vectors in $L_A$ with Hamming weight $2$, so $|D|=(q-1)\binom{s}{2}$ since $\Wt(\bm{a})=\Wt(\lambda\bm{a})$ if and only if $\lambda \in\Fq^*$ for $\alpha \neq \bm{0}$. As such, the length of $\mathcal{C}_{D^c}$ is $n=|D^c|=\frac{q^m-1}{q-1}-(q-1)\binom{s}{2}$. 
		
		Our next goal is to determine the value of $\min\limits_{\bm{x}\in (\Fq^m)^*}\{\sum\limits_{\bm{d}\in D}\sum\limits_{y\in\Fq^*}\omega_p^{\Tr(y\bm{x\cdot d})}\}$.
		For any $\bm{x}\in (\Fq^{m})^*$, we have
		\begin{align}\notag
			\sum_{\bm{d}\in D}\sum_{y\in\Fq^*}\omega_p^{\Tr(y\bm{x\cdot d})}&=\sum_{\bm{d}\in \Fq^s,\Wt(\bm{d})=2}\omega_p^{\Tr(\bm{x}_A\cdot\bm{d})}\\
         \notag
			&=K_2(a;s,q)\\
			&=\frac{q^2}{2}a^2-\left(\frac{2(q-1)qs+q(2-q)}{2}\right)a+\binom{s}{2}(q-1)^2, \label{eq:K_2explicit}
		\end{align}
	where $a=\Wt(\bm{x}_A)$.
		Its axis of symmetry is $a=\frac{2qs-2s+2-q}{2q}=s-\frac{1}{2}+\frac{1-s}{q}\in [1,s]$. Note that $a=\Wt(\bm{x}_A) \in \{0\}\cup[s]$, thus the above function achieves its minimum value at $a=\nearest{s-\frac{1}{2}+\frac{1-s}{q}}$. Therefore, we have 
  \begin{align*}
      \Delta:=&\min\limits_{\bm{x}\in (\Fq^{m})^*}\left\{\sum\limits_{\bm{d}\in D}\sum\limits_{y\in\Fq^*}\omega_p^{\Tr(y\bm{x\cdot d})}\right\}=K_2(\nearest{s-\frac{1}{2}+\frac{1-s}{q}};s,q).\\
  \end{align*}
		
  By Eq. \eqref{eq:calcu_minimum_weight}, we have
		\begin{align*}
			d:=\min_{\bm{x}\in(\Fq^{m})^*}\Wt(\bm{c_x})=\frac{(q-1)n+1}{q}+\frac{\Delta}{q}.\\
		\end{align*}
Now we prove that $d\textgreater 0$. For any $\bm{x}=(x_1,x_2,\dots,x_m)\in (\Fq^{m})^*$, we assume $x_{i_0}\neq 0$. Let $d_i=0\in \Fq$ for all $i\in [m]\backslash \{i_0\}$, and $d_{i_0}=1$.
 Then $\bm{d'}=(d_1,d_2,\dots,d_m)\in PG(m-1,q)$ satisfies $\Wt(\bm{d'})=1$ and $\bm{x}\cdot \bm{d'}\neq 0$, which means that $\bm{d'}\in D^c$ and then $\bm{c_x}=(\bm{x}\cdot \bm{d})_{\bm{d}\in D^c}\neq \bm{0}$. So we have $d=\min\limits_{\bm{x}\in(\Fq^{m})^*}\Wt(\bm{c_x})\textgreater 0$.
 Thus, the dimension and minimum distance of $\mathcal{C}_{D^c}$ are $m$ and $d$, respectively.

		By Theorem \ref{suffi_condi_lrc}, the code has locality $(2,q)$. Using the Plotkin bound, we obtain that
		\begin{align*}
			k_{opt}^{(q)}(n-q-1,d)&\leq\left\lfloor \log_q\frac{qd}{qd-(q-1)(n-q-1)}\right\rfloor=\left\lfloor \log_q\frac{qd}{\Delta+q^2}\right\rfloor\leq m-2.
		\end{align*}
		Thus, the code $\mathcal{C}_{D^c}$ is a $k$-optimal $(2,q)$-LRC according to the generalized C-M bound \eqref{CMbound_rdeltaLRC} with $\tau=1$.
	\end{proof}
	\begin{remark}
		By the techniques of graph theory and anti-codes, the authors in \cite[Theorem 10]{Anticode} have obtained these codes as 2-LRCs. And they only proved that these codes are $k$-optimal 2-LRCs for $s=3,m\geq 3,2\leq q\leq 14$. However, in Theorem \ref{thm:weight2}, we show that there are actually infinitely many  $k$-optimal $q$-ary $(2,q)$-LRCs in \cite[Theorem 10]{Anticode} for arbitrary prime power $q$.
	\end{remark}
	\begin{remark}
		In Theorem \ref{thm:weight2}, the condition $\binom{s}{2}(q-1)\leq\frac{q^{m-1}-q}{q-1}$ can be satisfied by taking $m\geq 1+\left\lceil\log_q\left(\binom{s}{2}(q-1)^2+q\right)\right\rceil$ and the condition $0\textless \frac{qd}{\Delta+q^2}\textless q^{m-1}$ can be satisfied by taking $s\leq 2q-\frac{(q-2)^2}{4(q-1)}$, respectively. The former is easy to verify and the reason for the latter is as follows. If $s\leq 2q-\frac{(q-2)^2}{4(q-1)}$, then we have $\Delta+q^2\geq K_2(s-\frac{1}{2}+\frac{1-s}{q};s,q)+q^2=-\frac{4s(q-1)+(q-2)^2}{8}+q^2\geq q$. Note that $d\textless q^{m-1}$ since $\mathcal{C}_{D^c}$ is a punctured simplex code and simplex codes are Griesmer codes. We obtain $0\textless \frac{qd}{\Delta+q^2}\textless q^{m-1}$.
	\end{remark}
    \begin{example}
        Let $q=2$, $m=4$ and $A=\{1,2,3,4\}$, then the code $\mathcal{C}_{D^c}$ defined in Theorem \ref{thm:weight2} is a $2$-ary $2$-LRC with parameters $[9,4,4]$. Its generator matrix is as follows:
        \begin{align*}
         G=\left[\begin{array}{ccccccccc}
             0& 0& 0& 0& 1& 1& 1& 1& 1\\
              0& 0& 1& 1& 0& 0& 1& 1& 1\\
              0& 1& 0& 1& 0& 1& 0& 1& 1\\
              1& 0& 0& 1& 0& 1& 1& 0& 1\\
        \end{array}\right].
        \end{align*}
    By the Plotkin bound, $k^{(2)}_{
opt}(6, 4) \leq  \left\lfloor \log_2 (\frac{8}{2})\right\rfloor = 2$. Hence $\mathcal{C}_{D^c}$ is a $k$-optimal $2$-LRC achieving the generalized C-M bound.
    \end{example}
	By adding all the points in $P_A$ with Hamming weight $1$ into the deleting set $D$, we have the following construction. It is motivated by \cite[Theorem 8]{LuoConstrkoptimalBin}.
	\begin{theorem}\label{thm:wt1and2} 
		Let $m\geq 4$ be a positive integer and $A\subseteq [m]$, where $|A|=s\geq 3$. Denote $D=\{\bm{d}\in P_A: \Wt(\bm{d})\in [2]\},D^c=P_{[m]}\backslash D$. Then the code $\mathcal{C}_{D^c}$ defined in Eq. \eqref{eq:CodeDef} is a $q$-ary $k$-optimal $(2,q)$-LRC with parameters $$[n,k,d]=\left[\frac{q^m-1}{q-1}-(q-1)\binom{s}{2}-s,m,\frac{(q-1)n+1}{q}+\frac{\Delta}{q}\right]$$
		providing that $\binom{s}{2}(q-1)+s\leq\frac{q^{m-1}-q}{q-1}$ and $0\textless \frac{qd}{\Delta+q^2}\textless q^{m-1}$, where $\Delta=K_1(\nearest{s-\frac{1}{2}+\frac{2-s}{q}};s,q)+K_2(\nearest{s-\frac{1}{2}+\frac{2-s}{q}};s,q)$.
	\end{theorem}
	\begin{proof}
		There are $(q-1)^2\binom{s}{2}$ vectors in $L_A$ with Hamming weight $2$ and $(q-1)s$ vectors in $L_A$ with Hamming weight $1$, so $|D|=(q-1)\binom{s}{2}+s$. Thus, the length of $\mathcal{C}_{D^c}$ is $n=|D^c|=\frac{q^m-1}{q-1}-|D|=\frac{q^m-1}{q-1}-(q-1)\binom{s}{2}-s$. 
  
  Next, we determine the value of $\min\limits_{\bm{x}\in (\Fq^{m})^*}\{\sum\limits_{\bm{d}\in D}\sum\limits_{y\in\Fq^*}\omega_p^{\Tr(y\bm{x\cdot d})}\}$.
		For any $\bm{x}\in (\Fq^{m})^*$, we have
		\begin{align*}
			\sum_{\bm{d}\in D}\sum_{y\in\Fq^*}\omega_p^{\Tr(y\bm{x\cdot d})}&=\sum_{\bm{d}\in \Fq^s,\Wt(\bm{d})=1,2}\omega_p^{\Tr(\bm{x}_A\cdot\bm{d})}\\
			&=K_1(a;s,q)+K_2(a;s,q)\\
			&=\frac{q^2}{2}a^2-\left(\frac{2(q-1)qs+q(2-q)+2q}{2}\right)a+\binom{s}{2}(q-1)^2+s(q-1),
		\end{align*}
	where $a=\Wt(\bm{x}_A)$. 	Its axis of symmetry is $a=\frac{2qs-2s+4-q}{2q}=s-\frac{1}{2}+\frac{2-s}{q}\in [1,s]$. Note that $a=\Wt(\bm{x}_A) \in \{0\}\cup[s]$, thus the above function takes its minimum value at $a=\nearest{s-\frac{1}{2}+\frac{2-s}{q}}$. Therefore, we have 
  \begin{align*}
      \Delta:&=\min\limits_{\bm{x}\in (\Fq^{m})^*}\left\{\sum\limits_{\bm{d}\in D}\sum\limits_{y\in\Fq^*}\omega_p^{\Tr(y\bm{x\cdot d})}\right\}\\
      &=K_1(\nearest{s-\frac{1}{2}+\frac{2-s}{q}};s,q)+K_2(\nearest{s-\frac{1}{2}+\frac{2-s}{q}};s,q).\\
  \end{align*}
		By Eq. \eqref{eq:calcu_minimum_weight}, we have
		\begin{align*}
			d:=&\min_{\bm{x}\in(\Fq^{m})^*}\Wt(\bm{c_x})=\frac{(q-1)n+1}{q}+\frac{\Delta}{q}.\\
		\end{align*}
Now we prove that $d\textgreater 0$. For any $\bm{x}=(x_1,x_2,\dots,x_m)\in (\Fq^{m})^*$, we assume $x_{i_0}\neq 0$. Let $d_i=1\in \Fq$ for all $i\in [m]\backslash \{i_0\}$, and $d_{i_0}=1-x_{i_0}^{-1}(\sum\limits_{i=1\atop i\neq i_0}^{m}x_id_i)$.
 Then $\bm{d'}=(d_1,d_2,d_3,\dots,d_m)\in PG(m-1,q)$ satisfies $\Wt(\bm{d'})\geq m-1\textgreater 2$ and $\bm{x}\cdot \bm{d'}\neq 0$, which means that $\bm{d'}\in D^c$ and then $\bm{c_x}=(\bm{x}\cdot \bm{d})_{\bm{d}\in D^c}\neq \bm{0}$. So we have $d=\min\limits_{\bm{x}\in(\Fq^{m})^*}\Wt(\bm{c_x})\textgreater 0$.
 Thus, the dimension and minimum distance of $\mathcal{C}_{D^c}$ are $m$ and $d$, respectively.
 
		According to Theorem \ref{suffi_condi_lrc}, the code has locality $(2,q)$. Using the Plotkin bound, we obtain that
		\begin{align*}
			k_{opt}^{(q)}(n-q-1,d)\leq\left\lfloor \log_q\frac{qd}{qd-(q-1)(n-q-1)}\right\rfloor=\log_q\left\lfloor\frac{qd}{\Delta+q^2}\right\rfloor\leq m-2.
		\end{align*}
	 By the generalized C-M bound \eqref{CMbound_rdeltaLRC}, taking $\tau=1$, we have 
  \[m=k \leq 2+k_{opt}^{(q)}(n-q-1,d)=m.\]
  Thus the code $\mathcal{C}_{D^c}$ is a $k$-optimal $(2,q)$-LRC.
	\end{proof}
\begin{remark}
    In \cite[Theorem 8]{LuoConstrkoptimalBin}, Luo and Cao showed that the code $\mathcal{C}_{D^c}$ in Theorem \ref{thm:wt1and2} are $k$-optimal $2$-LRCs for $q=2$ and $m,s$ satisfying $s=3,4$, $m\textgreater s$. In Theorem \ref{thm:wt1and2}, we generalize this result to arbitrary finite field $\Fq$ and obtain new families of $q$-ary $k$-optimal $(2,q)$-LRCs.
\end{remark}
\begin{remark}
In Theorem \ref{thm:wt1and2}, the condition $\binom{s}{2}(q-1)+s\leq\frac{q^{m-1}-q}{q-1}$ can be satisfied by taking $m\geq 1+\left\lceil\log_q(q+\binom{s}{2}(q-1)^2+s(q-1))\right\rceil$ and the condition $0\textless \frac{qd}{\Delta+q^2}\textless q^{m-1}$ can be satisfied by taking $s \leq 2q-\frac{(q-4)^2}{4(q-1)}$, respectively. The former is easy to verify and the reason for the latter is as follows. If $s\leq 2q-\frac{(q-4)^2}{4(q-1)}$, then we have $\Delta+q^2\geq K_1(s-\frac{1}{2}+\frac{2-s}{q};s,q)+K_2(s-\frac{1}{2}+\frac{2-s}{q};s,q)+q^2=-\frac{4s(q-1)+(q-2)^2}{8}+\frac{q}{2}-\frac{3}{2}+q^2\geq q$. Note that $qd\textless q^m$ since $\mathcal{C}_{D^c}$ is a punctured simplex code and simplex codes are Griesmer codes. We have $0\textless  \frac{qd}{\Delta+q^2}\textless q^{m-1}$.
\end{remark}
       \begin{example}
        Let $q=2$, $m=5$ and $A=\{1,2,3,4\}$, then the code $\mathcal{C}_{D^c}$ defined in Theorem \ref{thm:wt1and2} is a $2$-ary $2$-LRC with parameters $[21,5,10]$. Its generator matrix is as follows:
        \begin{align*}
         G=\left[\begin{array}{ccccccccccccccccccccc}
              0& 0& 0& 0& 1& 1& 1& 1& 1& 1& 0& 1& 1& 1& 1& 1 &0 &1  &0  &0 &0   \\
              0& 1& 1& 1& 0& 0& 0& 1& 1& 1& 1& 0& 1& 1& 1& 1 &0 &0  &1  &0 &0   \\
              1& 0& 1& 1& 0& 1& 1& 0& 0& 1& 1& 1& 0& 1& 1& 1 &0 &0  &0  &1 &0   \\
              1& 1& 0& 1& 1& 0& 1& 0& 1& 0& 1& 1& 1& 0& 1& 1 &0 &0  &0  &0 &1   \\
              1& 1& 1& 0& 1& 1& 0& 1& 0& 0& 1& 1& 1& 1& 0& 1 &1 &1  &1  &1 &1   \\
        \end{array}\right].
        \end{align*}
        By the Plotkin bound, $k^{(2)}_{
opt}(18, 10) \leq  \left\lfloor \log_2 (\frac{20}{2})\right\rfloor = 3$. Hence $\mathcal{C}_{D^c}$ is a $k$-optimal $2$-LRC achieving the generalized C-M bound.
    \end{example}
	In the following, we present another family of punctured simplex codes that are not only Griesmer codes but can also be $k$-optimal $(2,q)$-LRCs under certain constraints.
	\begin{theorem}\label{thm:except1} 
		Let $m\geq 3, A \subseteq [m],2\leq |A|=s\leq m-1, D=\{\bm{d}\in P_{A}: \Wt(\bm{d})\in [s]\backslash \{1\}\},D^c=P_{[m]}\backslash D$. Then the code $\mathcal{C}_{D^c}$ defined in Eq. \eqref{eq:CodeDef} is a $q$-ary linear code with parameters $[n,k,d]=\left[\frac{q^m-1}{q-1}-\frac{q^{s}-1}{q-1}+s,m,d=q^{m-1}-q^{s-1}+1\right]$. And the \Review{following hold}: 
		\begin{itemize}
			\item[(1)] $\mathcal{C}_{D^c}$ is a Griesmer code.
			\item[(2)] $\mathcal{C}_{D^c}$ is a $q$-ary $k$-optimal $(2,q)$-LRC.
		\end{itemize}
	\end{theorem}
	\begin{proof}
		There are $(q-1)s$ vectors in $L_A$ with Hamming weight $1$, so $|D|=\frac{q^s-1}{q-1}-\frac{(q-1)s}{q-1}=\frac{q^s-1}{q-1}-s$. Thus, the length of $\mathcal{C}_{D^c}$ is $n=\frac{q^m-1}{q-1}-|D|=\frac{q^m-1}{q-1}-\frac{q^s-1}{q-1}+s$.
  
  Next, we determine the value of $\min\limits_{\bm{x}\in (\Fq^{m})^*}\{\sum\limits_{\bm{d}\in D}\sum\limits_{y\in\Fq^*}\omega_p^{\Tr(y\bm{x\cdot d})}\}$.
		For any $\bm{x}\in (\Fq^{m})^*$, denote $\Wt(\bm x_A)=a$ and we have
		\begin{align*}
			\sum_{\bm{d}\in D}\sum_{y\in\Fq^*}\omega_p^{\Tr(y\bm{x\cdot d})}&=\sum_{\bm{d}\in (\Fq^s)^*,\Wt(\bm{d})\neq 1}\omega_p^{\Tr(\bm{x}_A\cdot\bm{d})}\\
			&=\sum_{\bm{d}\in (\Fq^s)^*}\omega_p^{\Tr(\bm{x}_A\cdot\bm{d})}-\sum_{\bm{d}\in \Fq^s,\Wt(\bm{d})= 1}\omega_p^{\Tr(\bm{x}_A\cdot\bm{d})}\\
			&=\begin{cases}
				q^s-1-K_1(0,s,q), & \text{ if }a=\Wt(\bm{x}_A)=0,\\
				-1-K_1(a,s,q), & \text{ if }a=\Wt(\bm{x}_A)\neq 0,\\
			\end{cases}\\
			&=\begin{cases}
				q^s-1-(q-1)s, & \text{ if }a=\Wt(\bm{x}_A)=0,\\
				-1-(q-1)s+qa, & \text{ if }a=\Wt(\bm{x}_A)\neq 0.\\
			\end{cases}\\
		\end{align*}
	In particular, $a$ can not be $0$ when $s=m$.
		It is easy to verify that its minimum value is $(q-1)(1-s)$, which is attained at $a=1$. By Eq. \eqref{eq:calcu_minimum_weight}, we have
		\begin{align*}
			d:=\min_{\bm{x}\in(\Fq^{m})^*}\Wt(\bm{c_x})=\frac{q-1}{q}n+\frac{1}{q}+\frac{(q-1)(1-s)}{q}=q^{m-1}-q^{s-1}+1\textgreater 0.\\
		\end{align*}
  Thus, the dimension and minimum distance of $\mathcal{C}_{D^c}$ are $m$ and $d$, respectively.  

  (1): Then the following equality holds: 
	\begin{align*}
			\sum_{i=0}^{m-1}\left\lceil\frac{d}{q^i}\right\rceil=\sum_{i=0}^{m-1}\left\lceil\frac{q^{m-1}-q^{s-1}+1}{q^i}\right\rceil=\frac{q^m-1}{q-1}-\frac{q^s-1}{q-1}+s=n,
		\end{align*} which implies that $\mathcal{C}_{D^c}$ is a Griesmer code. 
  
  (2): Note that $|D|=\frac{q^s-1}{q-1}-s\leq \frac{q^{m-1}-q}{q-1}$ since $2\leq s\leq m-1$. Thus, the code $\mathcal{C}_{D^c}$  has locality $(2,q)$ according to Theorem \ref{suffi_condi_lrc}. From $\sum_{i=0}^{m-2}\left\lceil\frac{d}{q^i}\right\rceil=n-1\textgreater n-q-1$ and the Griesmer bound, we obtain that $k_{opt}^{(q)}(n-q-1,d)\leq m-2.$
  By the generalized C-M bound \eqref{CMbound_rdeltaLRC}, taking $\tau=1$, we have 
  \[m=k \leq 2+k_{opt}^{(q)}(n-q-1,d)=m.\]
  Thus the code $\mathcal{C}_{D^c}$ is a $k$-optimal $(2,q)$-LRC.
	\end{proof}
 \begin{example}
     Let  $q=3, m=4$ and $A=\{1,2,3\}$. Then the code $\mathcal{C}_{D^c}$ defined in Theorem \ref{thm:except1} is a $3$-ary $(2,3)$-LRC with parameters $[30,4,19]$. Its generator matrix is as follows:
     \begin{align*}
         G=\left[\begin{array}{cccccccccccccccccccccccccccccc}
0&0&0&0&0&0&0&0&0&0&0&1&1&1&1&1&1&1&1&1&1&1&1&1&1&1&1&1&1&1\\
0&0&0&0&1&1&1&1&1&1&1&0&0&0&0&0&0&0&1&1&1&1&1&1&2&2&2&2&2&2\\
0&1&1&1&0&0&0&1&1&2&2&0&0&0&1&1&2&2&0&0&1&1&2&2&0&0&1&1&2&2\\
1&0&1&2&0&1&2&1&2&1&2&0&1&2&1&2&1&2&1&2&1&2&1&2&1&2&1&2&1&2\\
         \end{array}\right].
     \end{align*} 
$\mathcal{C}_{D^c}$ is both a Griesmer code and a $k$-optimal $(2,3)$-LRC. 
\end{example}
	In the following theorem, we generalize the above result and get $q$-ary $k$-optimal $(2,q)$-LRCs with new parameters.
	\begin{theorem}\label{thm:except1multi}
		Let $m\geq 3$, $t\geq 1$ and $A_1,A_2,\dots,A_t \subseteq [m]$ with $|A_i|=s_i\geq 2 \text{ for all }  i\in [t]$. Denote $D_i=\{\bm{d}\in P_{A_i}: \Wt(\bm{d})\in [s]\backslash \{1\}\}, D=\bigcup\limits_{i=1}^{t}D_i$ and $ 
 D^c=P_{[m]}\backslash D$. Suppose $|A_i\cap \bigcup\limits_{j=1\atop j\neq i}^{t}A_j|\leq 1$ for all $i \in [t]$. Then the code $\mathcal{C}_{D^c}$ defined in Eq. \eqref{eq:CodeDef} is a $q$-ary linear code with parameters $[n,k,d]=[\frac{q^m-1}{q-1}-\sum_{i=1}^{t}\frac{q^{s_i}-1}{q-1}+\sum_{i=1}^t s_i,m,q^{m-1}-\sum_{i=1}^{t} q^{s_i-1}+t]$. 
 It is a $q$-ary $k$-optimal $(2,q)$-LRC providing that $\sum_{i=1}^{t}\frac{q^{s_i}-1}{q-1}-\sum_{i=1}^t s_i\leq\frac{q^{m-1}-q}{q-1}$ and one of the following conditions holds (or both): 
\begin{itemize}           
	\item[(1)]     $0\textless \frac{q^m-\sum_{i=1}^{t}q^{s_i}+tq}{q^2-\sum_{i=1}^t((q-1)(s_i-1))}\textless q^{m-1}$,
	\item[(2)]    $\sum_{i=0}^{m-2}\lceil\frac{d}{q^i}\rceil\textgreater n-q-1$.
\end{itemize}
	\end{theorem}
	\begin{proof}
		There are $(q-1)s_i$ vectors in $L_{A_i}$ with Hamming weight $1$, so $|D_i|=\frac{q^{s_i}-1}{q-1}-s_i$. Note that the Hamming weights of points in $D_i$ are greater than $1$ and $|A_i\cap \bigcup\limits_{j=1\atop j\neq i}^{t}A_j|\leq 1$ for any $i\in [t]$. We have $D_i\cap D_j=\varnothing$ for any $i\neq j\in [t]$. Thus, the code length of $\mathcal{C}_{D^c}$ is $n=\frac{q^m-1}{q-1}-\sum_{i=1}^{t}|D_i|=\frac{q^m-1}{q-1}-\sum_{i=1}^t\frac{q^{s_i}-1}{q-1}+\sum_{i=1}^t s_i$.
  
    Next, we determine the value of $\min\limits_{\bm{x}\in (\Fq^{m})^*}\{\sum\limits_{\bm{d}\in D}\sum\limits_{y\in\Fq^*}\omega_p^{\Tr(y\bm{x\cdot d})}\}$.
		For any $\bm{x}\in (\Fq^{m})^*$, we have
		\begin{align*}
			&\sum_{\bm{d}\in D}\sum_{y\in\Fq^*}\omega_p^{\Tr(y\bm{x\cdot d})} =\sum_{i=1}^{t}\sum_{\bm{d}\in (\Fq^{s_i})^*,\atop\Wt(\bm{d})\neq 1}\omega_p^{\Tr(\bm{x}_{A_i}\cdot\bm{d})}\\
			&=\sum_{i=1}^{t}\left(\sum_{\bm{d}\in (\Fq^{s_i})^*}\omega_p^{\Tr(\bm{x}_{A_i}\cdot\bm{d})}-\sum_{\bm{d}\in \Fq^{s_i},\Wt(\bm{d})= 1}\omega_p^{\Tr(\bm{x}_{A_i}\cdot\bm{d})}\right)\\
			&=\sum_{i=1}^{t}f_i(a_i),
		\end{align*}
		where $a_i=\Wt(x_{A_i})$ and 
		\begin{align*}
			f_i(a_i)&=\begin{cases}
				q^{s_i}-1-K_1(a_i,s_i,q), & \text{ if } a_i=0, \\
				-1-K_1(a_i,s_i,q), & \text{ if } a_i\textgreater 0,\\
			\end{cases}=\begin{cases}
				q^{s_i}-1-(q-1)s_i, & \text{ if } a_i=0,  \\
				-1-(q-1)s_i+qa_i, & \text{ if } a_i\textgreater 0.\\
			\end{cases}
		\end{align*}
	In particular, $a_1,a_2,\dots,a_t$ can not all be $0$ when $\bigcup\limits_{i=1}^{t} A_i=[m]$.
		By choosing $\bm{x}\in (\Fq^m)^*$ such that $a_1=a_2=\dots=a_t=1$, these $t$ functions $f_1,f_2,\dots,f_t$ can achieve their minimum values simultaneously. Thus, $\min\limits_{\bm{x}\in (\Fq^{m})^*}\{\sum\limits_{\bm{d}\in D}\sum\limits_{y\in\Fq^*}\omega_p^{\Tr(y\bm{x\cdot d})}\}=\sum_{i=1}^{t}\big((q-1)(1-s_i)\big)$. By Eq. \eqref{eq:calcu_minimum_weight}, we have 
		\begin{align*}
			d:=\min_{\bm{x}\in(\Fq^{m})^*}\Wt(\bm{c_x})
			=\frac{q-1}{q}n+\frac{1}{q}+\frac{\sum_{i=1}^{t}\big((q-1)(1-s_i)\big)}{q}
			=q^{m-1}-\sum_{i=1}^{t}q^{s_i-1}+t \textgreater 0.
		\end{align*}
 The last inequality is due to the fact that $\sum_{i=1}^{t}(s_i-1)\leq m-1$. Therefore, the dimension and minimum distance of $\mathcal{C}_{D^c}$ are $m$ and $d$, respectively. Since $\sum_{i=1}^{t}\frac{q^{s_i}-1}{q-1}-\sum_{i=1}^t s_i\leq\frac{q^{m-1}-q}{q-1}$, $\mathcal{C}_{D^c}$ has locality $(2,q)$ by Theorem \ref{suffi_condi_lrc}. We then show that the condition (1) and (2) both lead to the $k$-optimality of $\mathcal{C}_{D^c}$.

(1): Using the Plotkin bound, we obtain that
\begin{align*}
	k_{opt}^{(q)}(n-q-1,d)&\leq\left\lfloor \log_q\frac{qd}{qd-(q-1)(n-q-1)}\right\rfloor\\
	&=\Review{\left\lfloor\log_q\frac{q^m-\sum_{i=1}^{t}q^{s_i}+tq}{q^2-\sum_{i=1}^t((q-1)(s_i-1))}\right\rfloor}\\
	&\leq m-2.
\end{align*}
Thus, the code $\mathcal{C}_{D^c}$ is a $k$-optimal $(2,q)$-LRC according to the generalized C-M bound \eqref{CMbound_rdeltaLRC} with $\tau=1$.

(2) Since $\sum_{i=0}^{m-2}\lceil\frac{d}{q^i}\rceil\Review{\textgreater} n-q-1$, we have $k^{(q)}_{opt}(n-q-1, d)\leq  m-2$ by the Griesmer
bound. Thus, the code $\mathcal{C}_{D^c}$ is a $k$-optimal $(2, q)$-LRC according to the generalized C-M
bound \eqref{CMbound_rdeltaLRC} with $\tau = 1$.
	\end{proof}
 \begin{example}
    Let $q=3,m=4$, $A_1=\{1,2,3\},A_2=\{3,4\}$. Then the code $\mathcal{C}_{D^c}$ defined in Theorem \ref{thm:except1multi} is a $3$-ary $(2,3)$-LRC with parameters $[28,4,17]$. Its generator matrix is as follows:
\begin{align*}
    G=\left[
    \begin{array}{cccccccccccccccccccccccccccc}
 0  &0  &0  &0  &0  &0  &0  &0  &0  &1  &1  &1  &1  &1  &1  &1  &1  &1  &1  &1  &1  &1  &1  &1  &1  &1  &1  &1 \\
 0  &0  &1  &1  &1  &1  &1  &1  &1  &0  &0  &0  &0  &0  &0  &0  &1  &1  &1  &1  &1  &1  &2  &2  &2  &2  &2  &2 \\
 0  &1  &0  &0  &0  &1  &1  &2  &2  &0  &0  &0  &1  &1  &2  &2  &0  &0  &1  &1  &2  &2  &0  &0  &1  &1  &2  &2 \\
 1  &0  &0  &1  &2  &1  &2  &1  &2  &0  &1  &2  &1  &2  &1  &2  &1  &2  &1  &2  &1  &2  &1  &2  &1  &2  &1  &2 \\
    \end{array}
    \right].
\end{align*}
    By the Plotkin bound, $k^{(3)}_{
opt}(24, 17) \leq  \left\lfloor \log_3 (\frac{51}{3})\right\rfloor=2$. Hence $\mathcal{C}_{D^c}$ is a $k$-optimal $(2, 3)$-
LRC achieving the generalized C-M bound.
    %
    %
 \end{example}
	In the following, we present two other new families of $q$-ary $k$-optimal $(2,q)$-LRCs, some of which are distance-optimal codes or Griesmer codes.
	\begin{theorem}\label{thm:s=3except2-complex} 
		Let $q\geq 3$, $m\geq 4$, $A\in [m]$ with $|A|=3.$ Denote $D=\{\bm{d}\in P_A: \Wt(\bm{d})\in [3]\backslash \{2\}\},D^c=P_{[m]}\backslash D$. Then the code $\mathcal{C}_{D^c}$ defined in Eq. \eqref{eq:CodeDef} is a $q$-ary $(2,q)$-LRC with parameters $[n,k,d]$,
		where $n=\frac{q^m-1}{q-1}-\frac{q^{3}-1}{q-1}+3(q-1)$,
		$k=m$, and 
		\begin{align*}
			d=\begin{cases}
       3^{m-1}-6,          &\text{ if } q=3,\\
       q^{m-1}-q^2+2q-2,    &\text{ if } q\geq 4.\\
			\end{cases}
		\end{align*}
  The optimality of $\mathcal{C}_{D^c}$ is as follows:
\begin{itemize}
    \item[(1)] The code $\mathcal{C}_{D^c}$ is a $q$-ary $k$-optimal $(2,q)$-LRC with respect to the generalized C-M bound.
    \item[(2)] If $q=3$, then the code $\mathcal{C}_{D^c}$ is distance-optimal.
    \item[(3)] If $q=4$, then the code $\mathcal{C}_{D^c}$ is a Griesmer code.
\end{itemize}
	\end{theorem}
	\begin{proof}
		There are $\binom{3}{2}(q-1)^2$ vectors in $L_A$ with Hamming weight $2$, so $|D|=\frac{q^3-1}{q-1}-\frac{3(q-1)^2}{q-1}=\frac{q^3-1}{q-1}-3(q-1)$. Thus, the length of $\mathcal{C}_{D^c}$ is $n=\frac{q^m-1}{q-1}-|D|=\frac{q^m-1}{q-1}-\frac{q^3-1}{q-1}+3(q-1)$.
  
  Next, we determine the value of $\min\limits_{\bm{x}\in (\Fq^{m})^*}\{\sum\limits_{\bm{d}\in D}\sum\limits_{y\in\Fq^*}\omega_p^{\Tr(y\bm{x\cdot d})}\}$.
		For any $\bm{x}\in (\Fq^{m})^*$, we have
		\begin{align*}
			\sum_{\bm{d}\in D}\sum_{y\in\Fq^*}\omega_p^{\Tr(y\bm{x\cdot d})}&=\sum_{\bm{d}\in (\Fq^3)^*,\Wt(\bm{d})\neq 2}\omega_p^{\Tr(\bm{x}_A\cdot\bm{d})}\\
			&=\sum_{\bm{d}\in (\Fq^3)^*}\omega_p^{\Tr(\bm{x}_A\cdot\bm{d})}-\sum_{\bm{d}\in \Fq^3,\Wt(\bm{d})= 2}\omega_p^{\Tr(\bm{x}_A\cdot\bm{d})}\\
			&=\begin{cases}
				q^3-1-K_2(a,3,q), & \text{ if }a=\Wt(\bm{x}_A)=0,\\
				-1-K_2(a,3,q), & \text{ if }a=\Wt(\bm{x}_A)\neq 0.\\
			\end{cases}\\
		\end{align*}
Note that $0 \leq a=wt(\bm x_A) \leq 3$. By comparing all the four values at $a=0,1,2,3$, we obtain that its minimum value is 
 \begin{align*}
 \Delta:=\begin{cases}
 -4, & (\text{ taking } a=3), \text{ if } q=3,\\  
 -q^2+4q-4, &    (\text{ taking } a=1), \text{ if } q\geq 4.\\
 \end{cases}
 \end{align*}
 
 Using Eq. \eqref{eq:calcu_minimum_weight}, we have 
		\begin{align*}
			d:=\min_{\bm{x}\in(\Fq^{m})^*}\Wt(\bm{c_x})=\frac{q-1}{q}n+\frac{1}{q}+\frac{\Delta}{q}=\begin{cases}
       3^{m-1}-6,             &\text{ if } q=3,\\
       q^{m-1}-q^2+2(q-1), &\text{ if } q\geq 4.\\
			\end{cases}
		\end{align*}
Note that $d\textgreater 0$, the dimension and minimum distance of $\mathcal{C}_{D^c}$ are $m$ and $d$, respectively.	
Since $|D|=\frac{q^3-1}{q-1}-3(q-1)\leq \frac{q^{m-1}-1}{q-1}-1$, the code has locality $(2,q)$ according to Theorem \ref{suffi_condi_lrc}. Using the Plotkin bound, we obtain that
		\begin{align*}
			k_{opt}^{(q)}(n-q-1,d)&\leq\left\lfloor \log_q\frac{qd}{qd-(q-1)(n-q-1)}\right\rfloor \leq \left\lfloor\log_q\frac{q^m}{q^2+\Delta}\right\rfloor \leq m-2.
		\end{align*}
		Thus, the code $\mathcal{C}_{D^c}$ is a $k$-optimal $(2,q)$-LRC according to the generalized C-M bound \eqref{CMbound_rdeltaLRC} with $\tau=1$.

  When $q=3$, the parameters of $\mathcal{C}_{D^c}$ are $[\frac{3^m-1}{2}-7,m,3^{m-1}-6]$. Note that 
  \begin{align*}
      \frac{3^m-1}{2}-7\textless \sum_{i=0}^{m-1}\left\lceil\frac{3^{m-1}-5}{3^i}\right\rceil=\frac{3^m-1}{2}-6.
  \end{align*}
By the Griesmer bound, there exists no $3$-ary linear code with parameters $[\frac{3^m-1}{2}-7,m,3^{m-1}-5]$. Thus, $\mathcal{C}_{D^c}$ is distance-optimal.

  When $q=4$, the parameters of $\mathcal{C}_{D^c}$ are $[\frac{4^m-1}{3}-12,m,4^{m-1}-10]$. Since 
  \begin{align*}
  \frac{4^m-1}{3}-12=\sum_{i=0}^{m-1}\left\lceil\frac{4^{m-1}-10}{4^i}\right\rceil,
  \end{align*}
  $\mathcal{C}_{D^c}$ is a Griesmer code.
\end{proof}
    \begin{example}
     Let $q=3, m=4$, and $A=\{1,2,3\}$. Then the code define in Theorem \ref{thm:s=3except2-complex} is a 3-ary $(2,3)$ LRC with parameters $(33,4,21)$. Its generator matrix is as follows:
\begin{align*}
G=\left[\begin{array}{ccccccccccccccccccccccccccccccccc}
0&0&0&0&0&0&0&0&0&0&0&1&1&1&1&1&1&1&1&1&1&1&1&1&1&1&1&1&1&1&1&1&1\\
0&0&0&1&1&1&1&1&1&1&1&0&0&0&0&0&0&0&0&1&1&1&1&1&1&1&2&2&2&2&2&2&2\\
0&1&1&0&0&1&1&1&2&2&2&0&0&1&1&1&2&2&2&0&0&0&1&1&2&2&0&0&0&1&1&2&2\\
1&1&2&1&2&0&1&2&0&1&2&1&2&0&1&2&0&1&2&0&1&2&1&2&1&2&0&1&2&1&2&1&2\\
\end{array}\right].
\end{align*}
By the Plotkin bound, $k^{(3)}_{
opt}(29, 21) \leq  \left\lfloor \log_3 (\frac{63}{5})\right\rfloor = 2$. Hence $\mathcal{C}_{D^c}$ is a $k$-optimal $(2, 3)$-
LRC achieving the generalized C-M bound.
    \end{example}
	\begin{theorem}\label{thm:s=4except2-complex}
		Let $m\geq 5$, $A_1,A_2\subseteq [m],$ $A=A_1\cup A_2$, where $|A_1|=|A_2|=2$ and $A_1\cap A_2=\varnothing$. Denote $D=P_A\backslash \{\bm{d}\in P_{A_1}\cup P_{A_2}: \Wt(\bm{d})=2\},$ and $ D^c=P_{[m]}\backslash D$. Then the code $\mathcal{C}_{D^c}$ defined in Eq. \eqref{eq:CodeDef} is a $q$-ary linear code with parameters $[n,k,d]$, 
		where $n=\frac{q^m-1}{q-1}-\frac{q^{4}-1}{q-1}+2(q-1)$,
		$k=m$,
		$d=
		q^{m-1}-q^{3}+q-2$. It is a $q$-ary $k$-optimal $(2,q)$-LRC provided $q\geq 3$. In addition, it is a Griesmer code when $q=3$.
	\end{theorem}
	\begin{proof}
		There are $\binom{2}{2}(q-1)^2$ vectors in $L_{A_i}$ with Hamming weight $2$, so $|D|=\frac{q^4-1}{q-1}-\frac{2(q-1)^2}{q-1}=\frac{q^4-1}{q-1}-2(q-1)$. Thus, the length of $\mathcal{C}_{D^c}$ is $n=\frac{q^m-1}{q-1}-|D|=\frac{q^m-1}{q-1}-\frac{q^4-1}{q-1}+2(q-1)$.
  
  Next, we determine the value of $\min\limits_{\bm{x}\in (\Fq^{m})^*}\{\sum\limits_{\bm{d}\in D}\sum\limits_{y\in\Fq^*}\omega_p^{\Tr(y\bm{x\cdot d})}\}$.
		For any $\bm{x}\in (\Fq^{m})^*$, we have
		\begin{align*}
			\sum_{\bm{d}\in D}\sum_{y\in\Fq^*}\omega_p^{\Tr(y\bm{x\cdot d})}&=\sum_{\bm{d}\in (\Fq^4)^*}\omega_p^{\Tr(\bm{x}_A\cdot\bm{d})}-\sum_{\bm{d}\in \Fq^2,\Wt(\bm{d})= 2}\omega_p^{\Tr(\bm{x}_{A_1}\cdot\bm{d})}-\sum_{\bm{d}\in \Fq^2,\Wt(\bm{d})= 2}\omega_p^{\Tr(\bm{x}_{A_2}\cdot\bm{d})}\\
			&=F(a_1,a_2),
		\end{align*}
		where $a_1=\Wt(\bm{x}_{A_1}),a_2=\Wt(\bm{x}_{A_2})$ and 
				\begin{align*}
					F(a_1,a_2)=\begin{cases}
						q^4-1-K_2(a_1;2,q)-K_2(a_2;2,q), &\text{ if } a_1=a_2=0,\\
						-1-K_2(a_1;2,q)-K_2(a_2;2,q), &\text{ otherwise. } 
					\end{cases}
				\end{align*}
		By comparing all the nine values at $(a_1,a_2)\in \{0,1,2\}^2$, its minimum value is $-q^2+2q-3$ (taking $a_1=2,a_2=0$). Using Eq. \eqref{eq:calcu_minimum_weight}, we have  
		\begin{align*}
			d&:=\min_{\bm{x}\in(\Fq^{m})^*}\Wt(\bm{c_x})=\frac{q-1}{q}n+\frac{1}{q}+\frac{-q^2+2q-3}{q}=q^{m-1}-q^3+q-2\textgreater 0.
		\end{align*}
        Thus, the dimension and minimum distance of $\mathcal{C}_{D^c}$ are $m$ and $d$, respectively.
		Since $\frac{q^4-1}{q-1}-2(q-1)\leq \frac{q^{m-1}-1}{q-1}-1$, the code $\mathcal{C}_{D^c}$ has locality $(2,q)$ according to Theorem \ref{suffi_condi_lrc}. When $q\geq 3$, by using the Plotkin bound, we obtain that
		\begin{align*}
			k_{opt}^{(q)}(n-q-1,d)&\leq\left\lfloor \log_q\frac{qd}{qd-(q-1)(n-q-1)}\right\rfloor\notag\\
			&\leq \left\lfloor\log_q\frac{q^m-q^4+q^2-2q}{2q-3}\right\rfloor\\
			&\leq \left\lfloor\log_q\frac{q^m-q^4+q^2-2q}{q}\right\rfloor\\
			&\leq m-2.
		\end{align*}
		Thus, the code $\mathcal{C}_{D^c}$ is a $k$-optimal $(2,q)$-LRC according to the generalized C-M bound \eqref{CMbound_rdeltaLRC} with $\tau=1$. When $q=3$, we have 
        \begin{align*}
		   \sum_{i=0}^{m-1} \left\lceil\frac{d}{3^i}\right\rceil{=}\sum_{i=0}^{m-1} \left\lceil\frac{3^{m-1}-26}{3^i}\right\rceil{=}\sum_{i=0}^{3}(3^{\Review{m-i-1}}{-}\frac{27}{3^i}{+}1){+}\sum_{i=4}^{m-1}3^{\Review{m-i-1}}=\frac{3^m-1}{2}{-}36{=}n,  
		\end{align*}
  which means that $\mathcal{C}_{D^c}$ is a Griesmer code.
	\end{proof}
 \begin{example}
     Let $q=4, m=6$, $A_1=\{1,2\},A_2=\{3,4\}$. Then the code $\mathcal{C}_{D^c}$ defined in Theorem \ref{thm:s=4except2-complex} is a $(2,4)$-LRC with parameters $[1286,6,962]$.
By the Plotkin bound, $k^{(4)}_{opt}(1281,962) \leq \left\lfloor \log_4 
\frac{4*962}{5} \right\rfloor=4$. Hence $\mathcal{C}_{D^c}$ is a $k$-optimal $(2,4)$-LRC achieving the generalized C-M bound.
 \end{example}
	The following theorem generalizes the construction of Theorem \ref{thm:weight2}.
	\begin{theorem}\label{thm:wt2mult}
		Let $m\geq 3$, $t\geq 1$, $A_i\subseteq [m]$ with $|A_i|=s_i\geq 2$ for $i\in[t].$ Denote $D_i=\{\bm{d}\in P_{A_i}|\Wt(\bm{x})=2\}$ for $i\in [t]$, $D=\bigcup\limits_{i=1}^{t} D_i$ and $D^c=P_{[m]}\backslash D$. Suppose $|A_i\cap \bigcup\limits_{j=1\atop j\neq i}^{t}A_j|\leq 1$ for all $i\in[t]$.
		Then the code $\mathcal{C}_{D^c}$ defined in \eqref{eq:CodeDef} is a $q$-ary $k$-optimal $(2,q)$-LRC with parameters $$[n,k,d]=\left[\frac{q^m-1}{q-1}{-}(q-1)\sum_{i=1}^t\binom{s_i}{2},m,\frac{(q-1)n+1}{q}+\frac{\Delta}{q} \right]$$
		providing that $(q-1)\sum_{i=1}^t\binom{s_i}{2}\leq\frac{q^{m-1}-q}{q-1} $ and $0\textless\frac{qd}{q^2+\Delta}\textless q^{m-1}$, where $\Delta=\sum_{i=1}^{t}K_2(\nearest{s_i-\frac{1}{2}+\frac{1-s_i}{q}};s_i,q)$.
	\end{theorem}
	\begin{proof}
		For all $i\neq j\in[t]$, $D_i\cap D_j=\{\bm{d}\in P_{A_i\cap A_j}|\Wt(\bm{d})=2\}=\varnothing$ since $|A_i\cap A_{i'}|\leq |A_i\cap \bigcup\limits_{j=1 \atop j\neq i}^{t}A_{j}|\leq 1$ for any $i\neq i' \in [t]$. So we have $|D|=\sum_{i=1}^t|D_i|=(q-1)\sum_{i=1}^t\binom{s_i}{2}$. Thus, the length of $\mathcal{C}_{D^c}$ is $n=\frac{q^m-1}{q-1}-|D|=\frac{q^m-1}{q-1}-(q-1)\sum_{i=1}^t\binom{s_i}{2}$. 
  
  Next, we determine the value of $\min\limits_{\bm{x}\in (\Fq^{m})^*}\{\sum\limits_{\bm{d}\in D}\sum\limits_{y\in\Fq^*}\omega_p^{\Tr(y\bm{x\cdot d})}\}$.
		For any $\bm{x}\in (\Fq^{m})^*$, we have
		\begin{align*}
			\sum_{\bm{d}\in D}\sum_{y\in \Fq^*}\omega_p^{\Tr(y\bm{d\cdot x})}&=\sum_{i=1}^t\sum_{\bm{d}\in D_i}\sum_{y\in \Fq^*}\omega_p^{\Tr(y\bm{d\cdot x})}\\
			&=\sum_{i=1}^t\sum_{\bm{d}\in\Fq^{s_i},\atop \Wt(\bm{d})=2}\omega_p^{\Tr(\bm{x}_{A_i}\cdot\bm{d})}\\
			&=\sum_{i=1}^{t} K_2(a_i;s_i,q),
		\end{align*}
		where $a_i=\Wt(\bm{x}_{A_i})$. In particular, $a_1,a_2,\dots,a_t$ can not all be $0$ when $s=m$. 
		Note that for $i\in [t]$, the axis of symmetry of $K_2(a_i;s_i,q)$ is $a_i=\frac{2(q-1)s_i+2-q}{2q}=s_i-\frac{1}{2}+\frac{1-s_i}{q}\in [1,s_i]$ and $a_i=\Wt(\bm{x}_{A_i})\in \{0\}\cup [s_i]$. 
		Meanwhile, $|A_i\cap \bigcup\limits_{j=1 \atop j\neq i}^{t}A_{j}|\leq 1$ for all $i\in[t]$. So it is possible for us to choose $\bm{x}\in (\Fq^{m})^*$
  such that $a_i=wt(\bm{x}_{A_i})=\nearest{s_i-\frac{1}{2}+\frac{1-s_i}{q}}$ for all $i\in [t]$, in which case
  $K_2(a_i;s_i,q)$ can get the minimum value simultaneously. Thus, we have \begin{align*}
      \Delta:=\min\limits_{\bm{x}\in (\Fq^{m})^*}\left\{\sum\limits_{\bm{d}\in D}\sum\limits_{y\in\Fq^*}\omega_p^{\Tr(y\bm{x\cdot d})}\right\}=\sum_{i=1}^{t} K_2(\nearest{s_i-\frac{1}{2}+\frac{1-s_i}{q}};s_i,q).
  \end{align*}
By \eqref{eq:calcu_minimum_weight}, we have
		\begin{align*}
			d&:=\min_{\bm{x}\in(\Fq^{m})^*}\Wt(\bm{c_x})=\frac{(q-1)n+1}{q}+\frac{\Delta}{q}.\\
		\end{align*}
Now we prove that $d\textgreater 0$. For any $\bm{x}=(x_1,x_2,\dots,x_m)\in (\Fq^{m})^*$, we assume $x_{i_0}\neq 0$. Let $d_i=0\in \Fq$ for all $i\in [m]\backslash \{i_0\}$, and $d_{i_0}=1$.
 Then $\bm{d'}=(d_1,d_2,\dots,d_m)\in PG(m-1,q)$ satisfies $\Wt(\bm{d'})=1$ and $\bm{x}\cdot \bm{d'}\neq 0$, which means that $\bm{d'}\in D^c$ and then $\bm{c_x}=(\bm{x}\cdot \bm{d})_{\bm{d}\in D^c}\neq \bm{0}$. So we have $d=\min\limits_{\bm{x}\in(\Fq^{m})^*}\Wt(\bm{c_x})\textgreater 0$.
 Thus, the dimension and minimum distance of $\mathcal{C}_{D^c}$ are $m$ and $d$, respectively.
		
	According to Theorem \ref{suffi_condi_lrc}, the code has locality $(2,q)$. Using the Plotkin bound, we obtain that
		\begin{align*}
			k_{opt}^{(q)}(n-q-1,d)&\leq\left\lfloor\log_q\frac{qd}{qd-(q-1)(n-q-1)}\right\rfloor\Review{=}\left\lfloor\log_q\frac{qd}{\Delta+q^2}\right\rfloor\leq m-2.
		\end{align*}
		Thus, the code $\mathcal{C}_{D^c}$ is a $k$-optimal $(2,q)$-LRC according to \eqref{CMbound_rdeltaLRC} with $\tau=1$.
	\end{proof}
\begin{remark}
	In Theorem \ref{thm:wt2mult}, the condition $\sum_{i=1}^t\left(\binom{s_i}{2}(q-1)\right)\leq\frac{q^{m-1}-q}{q-1}$ can be satisfied by taking $m\geq 1+\left\lceil\log_q\left(q+\sum_{i=1}^t\left(\binom{s_i}{2}(q-1)^2\right)\right)\right\rceil$ and the condition $0\textless \frac{qd}{\Delta+q^2}\textless q^{m-1}$ can be satisfied by taking $\sum_{i=1}^{t}s_i\Review{\leq} 2q-\frac{t(q-2)^2}{4(q-1)}$, respectively. The former is easy to verify and the reason for the latter is as follows. If $\sum_{i=1}^{t}s_i\leq 2q-\frac{t(q-2)^2}{4(q-1)}$, then we have $\Delta+q^2\geq \sum_{i=1}^t K_2(s_i-\frac{1}{2}+\frac{1-s_i}{q};s_i,q)+q^2=\sum_{i=1}^{t}-\frac{4s_i(q-1)+(q-2)^2}{8}+q^2\geq q$. Note that $qd\textless q^m$ since $\mathcal{C}_{D^c}$ is a punctured simplex code and simplex codes are Griesmer codes. We have $0\textless \frac{qd}{\Delta+q^2}\textless q^{m-1}$.
\end{remark}
		\begin{example}
			Let $q=2,m=5,A_1=\{1,2,3\}$ and $A_2=\{3,4,5\}$. Then the binary code $\mathcal{C}_{D^c}$ defined in Theorem \ref{thm:wt2mult} has parameters $[25,5,12]$ and the generator matrix is 
   $$G=\left[\begin{array}{ccccccccccccccccccccccccc}
		1 &0 &0 &1 &0 &1 &0 &1 &1 &0 &1 &0 & 1 &0 &1 &1 &0 &1 &0 &0 &1 &0 &1 &0 &1 \\
		0 &1 &0 &1 &0 &0 &1 &1 &0 &1 &1 &0 & 0 &1 &1 &0 &1 &1 &0 &1 &1 &0 &0 &1 &1 \\
		0 &0 &1 &1 &0 &0 &0 &0 &1 &1 &1 &0 & 0 &0 &0 &1 &1 &1 &0 &0 &0 &1 &1 &1 &1 \\
		0 &0 &0 &0 &1 &1 &1 &1 &1 &1 &1 &0 & 0 &0 &0 &0 &0 &0 &1 &1 &1 &1 &1 &1 &1 \\
		0 &0 &0 &0 &0 &0 &0 &0 &0 &0 &0 &1 & 1 &1 &1 &1 &1 &1 &1 &1 &1 &1 &1 &1 &1
	\end{array}\right].$$
			 By the Plotkin bound, $k^{(2)}_{opt}(22,12) \leq \left\lfloor \log_2 (\frac{24}{2}) \right\rfloor=3$. Hence $\mathcal{C}_{D^c}$ is a $k$-optimal $2$-LRC achieving the generalized C-M bound.
	\end{example}
 
	The following theorem generalizes the construction of Theorem 
 \ref{thm:wt1and2}.
 
	\begin{theorem}\label{thm:wt12mult} 
		Let $m\geq 4$, $t\geq 1$, $A_i\subseteq [m]$ with $|A_i|=s_i\geq 3$ for $i\in[t].$ Denote $D_i=\{\bm{d}\in P_{A_i}|\Wt(\bm{x})=1,2\}$ for $i\in [t]$, $D=\bigcup\limits_{i=1}^{t} D_i,$ and $ D^c=P_{[m]}\backslash D$. Suppose $A_i\cap A_j= \varnothing$ for all $i\neq j\in[t]$.
		Then the code $\mathcal{C}_{D^c}$ defined in \eqref{eq:CodeDef} is a $q$-ary $k$-optimal $(2,q)$-LRC with parameters
		\begin{align*}
			[n,k,d]=\left[\frac{q^m-1}{q-1}-\sum_{i=1}^t\bigg((q-1)\binom{s_i}{2}+s_i\bigg),m,\frac{(q-1)n+1}{q}+\frac{\Delta}{q} \right]
		\end{align*}
		providing that $\sum_{i=1}^t\left((q-1)\binom{s_i}{2}+s_i\right)\leq\frac{q^{m-1}-q}{q-1}$ and $0\textless \frac{qd}{\Delta+q^2}\textless q^{m-1}$, where $\Delta=\sum_{i=1}^{t}\left(K_1(\nearest{s_i-\frac{1}{2}+\frac{2-s_i}{q}};s_i,q)+K_2(\nearest{s_i-\frac{1}{2}+\frac{2-s_i}{q}};s_i,q)\right)$. 
	\end{theorem}
	\begin{proof}
		For all $i\neq j\in[t]$, $D_i\cap D_j=\varnothing$ since $A_i\cap A_j=\varnothing$. $|D|=\sum_{i=1}^t|D_i|=(q-1)\sum_{i=1}^t\binom{s_i}{2}+\sum_{i=1}^{t}s_i$. Thus, the length of $\mathcal{C}_{D^c}$ is $n=\frac{q^m-1}{q-1}-|D|=\frac{q^m-1}{q-1}-\sum_{i=1}^t|D_i|=\frac{q^m-1}{q-1}-(q-1)\sum_{i=1}^t\binom{s_i}{2}-\sum_{i=1}^{t}s_i.$
  
  Next, we determine the value of $\min\limits_{\bm{x}\in (\Fq^{m})^*}\{\sum\limits_{\bm{d}\in D}\sum\limits_{y\in\Fq^*}\omega_p^{\Tr(y\bm{x\cdot d})}\}$.
		For any $\bm{x}\in (\Fq^{m})^*$, we have
		\begin{align*}
			\sum_{\bm{d}\in D}\sum_{y\in \Fq^*}&\omega_p^{\Tr(y\bm{d\cdot x})}
			=\sum_{i=1}^t\sum_{\bm{d}\in D_i}\sum_{y\in \Fq^*}\omega_p^{\Tr(y\bm{d\cdot x})}\\
			&=\sum_{i=1}^t\sum_{\bm{d}\in\Fq^{s_i},\Wt(\bm{d})=1,2}\omega_p^{\Tr(\bm{x}_{A_i}\cdot\bm{d})}\\
			&=\sum_{i=1}^{t} \bigg(K_1(a_i;s_i,q)+K_2(a_i;s_i,q)\bigg)\\
			&=\sum_{i=1}^{t}\bigg(\frac{q^2}{2}a_i^2-\big(\frac{2(q-1)qs_i+q(4-q)}{2}\big)a_i+\frac{s_i(s_i-1)}{2}(q-1)^2+s_i(q-1)\bigg),
		\end{align*}
		where $a_i=\Wt(\bm{x}_{A_i})$. 
		Note that for $i\in [t]$, the axis of symmetry of $K_1(a_i;s_i,q)+K_2(a_i;s_i,q)$ is $a_i=\frac{2(q-1)s_i+4-q}{2q}=s_i-\frac{1}{2}+\frac{2-s_i}{q}\in [1,s_i]$ and $a_i=\Wt(\bm{x}_{A_i})\in \{0\}\cup [s_i]$. 
		Meanwhile, $A_i\cap A_j=\varnothing$ for all $i\neq j\in[t]$. So there exists $\bm{x}\in (\Fq^m)^*$ such that $a_i=\Wt(\bm{x}_{A_i})=\nearest{s_i-\frac{1}{2}+\frac{2-s_i}{q}}$ for all $i\in [t]$, in which case $K_1(a_i;s_i,q)+K_2(a_i;s_i,q)$ can get the minimum value simultaneously. 
  Thus, we have 
\begin{align*}
\Delta:=\min\limits_{\bm{x}\in (\Fq^{m})^*}\left\{\sum\limits_{\bm{d}\in D}\sum\limits_{y\in\Fq^*}\omega_p^{\Tr(y\bm{x\cdot d})}\right\}
=\sum_{i=1}^{t}\sum_{w=1}^2\big(K_w(\nearest{s_i-\frac{1}{2}+\frac{2-s_i}{q}};s_i,q)\big).
\end{align*}
  By Eq. \eqref{eq:calcu_minimum_weight}, we have
		\begin{align*}
			d&:=\min_{\bm{x}\in(\Fq^{m})^*}\Wt(\bm{c_x})=\frac{(q-1)n+1}{q}+\frac{\Delta}{q}.\\
		\end{align*}
Now we prove that $d\textgreater 0$. For any $\bm{x}=(x_1,x_2,\dots,x_m)\in (\Fq^{m})^*$, we assume $x_{i_0}\neq 0$. Let $d_i=1\in \Fq$ for all $i\in [m]\backslash \{i_0\}$, and $d_{i_0}=1-x_{i_0}^{-1}(\sum\limits_{i=1\atop i\neq i_0}^{m}x_id_i)$.
 Then $\bm{d'}=(d_1,d_2,d_3,\dots,d_m)\in PG(m-1,q)$ satisfies $\Wt(\bm{d'})\geq m-1\textgreater 2$ and $\bm{x}\cdot \bm{d'}\neq 0$, which means that $\bm{d'}\in D^c$ and then $\bm{c_x}=(\bm{x}\cdot \bm{d})_{\bm{d}\in D^c}\neq \bm{0}$. So we have $d=\min\limits_{\bm{x}\in(\Fq^{m})^*}\Wt(\bm{c_x})\textgreater 0$.
 Thus, the dimension and minimum distance of $\mathcal{C}_{D^c}$ are $m$ and $d$, respectively.

	According to Theorem \ref{suffi_condi_lrc}, the code has locality $(2,q).$ Using the Plotkin bound, we obtain that
		\begin{align*}
			k_{opt}^{(q)}(n-q-1,d)&\leq\left\lfloor\log_q\frac{qd}{qd-(q-1)(n-q-1)}\right\rfloor\\
			&=\left\lfloor\log_q\frac{qd}{\Delta+q^2}\right\rfloor\\
			&\leq m-2.
		\end{align*}
		Therefore, the code $\mathcal{C}_{D^c}$ is a $k$-optimal $(2,q)$-LRC according to \eqref{CMbound_rdeltaLRC} with $\tau=1$.
	\end{proof}
 
	\begin{remark}
		In Theorem \ref{thm:wt12mult}, the condition $\sum_{i=1}^t\left(\binom{s_i}{2}(q-1)+s_i\right)\leq\frac{q^{m-1}-q}{q-1}$ can be satisfied by taking $m\geq 1+\left\lceil\log_q\left(q+\sum_{i=1}^t\left(\binom{s_i}{2}(q-1)^2+s_i(q-1)\right)\right)\right\rceil$ and the condition $0\textless \frac{qd}{\Delta+q^2}\textless q^{m-1} $ can be satisfied by taking $\sum_{i=1}^{t}s_i\leq 2q-\frac{t(q-4)^2}{4(q-1)}$, respectively. The former is easy to verify and the reason for the latter is as follows. If $\sum_{i=1}^{t}s_i\leq 2q-\frac{t(q-4)^2}{4(q-1)}$, then we have $\Delta+q^2\geq \sum_{i=1}^{t}\big(K_1(s_i-\frac{1}{2}+\frac{2-s_i}{q};s_i,q)+K_2(s_i-\frac{1}{2}+\frac{2-s_i}{q};s_i,q)\big)+q^2=\sum_{i=1}^{t}\left(-\frac{4s_i(q-1)+(q-2)^2}{8}+\frac{q}{2}-\frac{3}{2}\right)+q^2\geq q$. Note that $qd\textless q^m$ since $\mathcal{C}_{D^c}$ is a punctured simplex code and simplex codes are Griesmer codes. We have $0\textless \frac{qd}{\Delta+q^2}\textless q^{m-1}$.
	\end{remark}
 
    \begin{example}
        Let $q=4, m=8$, $A_1=\{1,2,3\}$ and $A_2:=\{4,5,6,7,8\}$. Then the code $\mathcal{C}_{D^c}$ defined in Theorem \ref{thm:wt12mult} is a $4$-ary $(2,4)$-LRC with parameters $[21798,8,16346]$. By the Plotkin bound, $k^{(4)}_{
opt}(21793, 16346) \leq  \left\lfloor \log_4 (\frac{4*16346}{5})\right\rfloor = 6$. Hence $\mathcal{C}_{D^c}$ is a $k$-optimal $(2, 4)$-LRC achieving the generalized C-M bound.
    \end{example}
	\section{Concluding Remarks}
	\label{sec:conclusion}
	In this paper, we have investigated new constructions of optimal $(2, \delta)$-LRCs via punctured simplex codes. By using the language of finite geometry, we \Review{proposed} a simple but useful condition to ensure that a linear code has $(2,\delta)$-locality. It's worth noting that the locality of a punctured simplex code can be determined by the size of \Review{the puncturing set}. This allows us to have more flexible choices in constructing various families of locally repairable codes. According to some character sums and Krawtchouk polynomials,  we \Review{obtained} several infinite families of $q$-ary $(2,\delta)$-LRCs. All these codes are optimal with respect to the generalized C-M bound. We not only \Review{generalized} some previous results of 2-LRCs to the $(2, \delta)$-LRCs, but also \Review{constructed} some new optimal $(2, \delta)$-LRCs which are not optimal in the sense of 2-LRCs. Surprisingly, some of these codes are also Griesmer codes or distance-optimal codes.
	
	It is interesting to find more new optimal $(2, \delta)$-LRCs and generalize these results to the $(r,\delta)$-LRCs with $r \geq 3$ in the future.

\bmhead{Acknowledgments}
This research is supported in part by National Key Research and Development Program of China under Grant Nos. 2021YFA1001000 and 2022YFA1004900, the National Natural Science Foundation of China under Grant Nos. 62201322, 12201362, 12001322 and 12231014, the Natural Science Foundation of Shandong Province under Grant Nos. ZR2022QA031 and ZR2021QA043, Taishan scholar program of Shandong Province.

\section*{Declarations}
\noindent
\textbf{Conflict of interest} There are no competing interests that are directly or indirectly related to the work submitted.

\noindent
\textbf{Availability of data and materials} Data sharing not applicable to this article as no datasets were generated or analyzed during the current study.

\bigskip






\bibliographystyle{splnc04}
\bibliography{bibliofile}


\end{document}